\documentclass[pdflatex,sn-basic]{sn-jnl}

\usepackage{graphicx}%
\usepackage{multirow}%
\usepackage{amsmath,amssymb,amsfonts}%
\usepackage{amsthm}%
\usepackage{mathrsfs}%
\usepackage[title]{appendix}%
\usepackage{xcolor}%
\usepackage{textcomp}%
\usepackage{manyfoot}%
\usepackage{booktabs}%
\usepackage{algorithm}%
\usepackage{algorithmicx}%
\usepackage{algpseudocode}%
\usepackage{listings}%
\usepackage{graphicx}
\usepackage{wrapfig}
\usepackage{placeins}
\usepackage{booktabs}
\usepackage{tabularx}
\usepackage{array}
\usepackage{float}

\theoremstyle{thmstyleone}%
\newtheorem{theorem}{Theorem}
\newtheorem{proposition}[theorem]{Proposition}%

\theoremstyle{thmstyletwo}%
\newtheorem{remark}{Remark}%

\theoremstyle{thmstylethree}%
\newtheorem{definition}{Definition}%
\newtheorem{hypothesis}{Hypothesis}
\newcommand{\Hh}{\mathcal{H}}

\newcommand{\C}{\mathbb{C}}
\newcommand{\R}{\mathbb{R}}

\newcommand{\dd}{\mathrm{d}}
\newcommand{\Atr}{\mathcal{A}_{\mathrm{tr}}}
\newcommand{\Loc}{\operatorname{Loc}}
\newcommand{\id}{\mathrm{id}}

\newcommand{\EP}{\mathcal{E}_{P}}
\newcommand{\Gad}{\mathscr{G}_{\mathrm{ad}}}
\newcommand{\Pad}{\mathscr{P}_{\mathrm{ad}}}

\newcommand{\Dom}{\operatorname{Dom}}
\newcommand{\HH}{\mathbb{H}}
\newcommand{\SU}{\mathrm{SU}}

\newcommand{\Sig}{\Sigma}

\begin{document}

\title[Quaternionic protein deformation response]{Quaternionic Response Geometry for Proteins: Toward a Noncommutative Theory of Ordered Deformations}


\author[1]{\fnm{Xiaoting} \sur{Chen}}

\author[1,4]{\fnm{Chon-Fai} \sur{Kam}}

\author[2]{\fnm{Yu} \sur{Li}}

\author[3]{\fnm{David} \sur{Medina-Ortiz}}

\author[4]{\fnm{Cedric} \sur{Damour}}

\author[4]{\fnm{Jean Pierre} \sur{Chabriat}}

\author[5]{\fnm{Alain} \sur{Miranville}}

\author[4]{\fnm{Miloud} \sur{Bessafi}}

\author*[1,6]{\fnm{Frederic} \sur{Cadet}}\email{
frederic.cadet.run@gmail.com}

\affil*[1]{\orgname{University Paris City \& University of Reunion}, \orgaddress{\city{Paris}, \postcode{75015}, \country{France}}}

\affil[2]{\orgdiv{School of Information Science and Technology and Beijing Institute of Artificial Intelligence}, \orgaddress{\city{Beijing}, \country{China}}}

\affil[3]{\orgdiv{Departamento de Ingeniería en Computación}, \orgname{Universidad de Magallanes}, \orgaddress{\city{Punta Arenas}, \country{Chile}}}

\affil[4]{\orgdiv{ENERGYLab}, \orgname{University of Reunion}, \orgaddress{\city{Saint-Denis}, \country{France}}}

\affil[5]{\orgdiv{Laboratoire de Mathématiques Appliquées du Havre (LMAH)}, \orgname{Université Le Havre Normandie}, \orgaddress{\city{Le Havre}, \country{France}}}

\affil*[6]{\orgdiv{PEACCEL}, \orgname{AI for Biologics}, \orgaddress{\city{Paris}, \postcode{} \country{France}}}


\abstract{
Protein function may depend not only on endpoint conformations but also
on the ordered deformation histories through which they are reached.
This distinction is relevant to allostery, conformational switching,
mutation-induced rearrangements, and epistatic effects, where different
perturbation sequences may produce similar visible structures while
retaining distinct internal transport histories. Current state-centered
or endpoint-centered representations do not always preserve this
order-sensitive information. The practical motivation is therefore to
provide a foundation for future descriptors of protein deformation
trajectories that can distinguish ordered histories even when endpoint
conformations are similar. Such descriptors could support analyses of
allosteric switching, mutation-order effects, conformational memory,
and path-dependent response in molecular-dynamics trajectories, NMR
ensembles, structural families, and outputs of geometric generative
models.

We propose a deformation-first geometric framework based on
quaternionic frame transport along the protein backbone. Local backbone
frames are lifted to quaternionic variables, with infinitesimal rotation
encoded by
\(
\Omega(\ell)=2\,q(\ell)^{-1}\partial_\ell q(\ell).
\)
Ordered concatenation of admissible deformation paths generates a
noncommutative transport algebra, recording that deformation \(A\)
followed by \(B\) need not be equivalent to \(B\) followed by \(A\).
From this ordered transport layer, we construct a spectral-response
layer comprising a global Dirac-type operator, local spectral germs, a
renormalized spectral density, and a mixed response form.

A minimal realization on an idealized \(\alpha\)-helix shows how
localized pitch and bending perturbations can yield similar endpoint
descriptors while producing a nonzero endpoint-derived ordered-transport discrepancy.
At the formal level, the
framework separates an order-sensitive transport-memory sector, lost
under a commutative shadow, from a spectral-response sector that remains
visible.
}


\maketitle

\section{Introduction}\label{sec:intro}

Proteins are not rigid molecular objects. Their biological function often depends on conformational ensembles, local rearrangements, allosteric communication, and dynamical response pathways \citep{thirumalai2019symmetry}. A ligand-binding event, mutation, change in local flexibility, or mechanical perturbation can alter not only the final conformation of a protein but also the pathway through which that conformation is reached \citep{corbella2023loop}.

\begin{figure}[t]
    \centering
    \includegraphics[width=\textwidth]
    {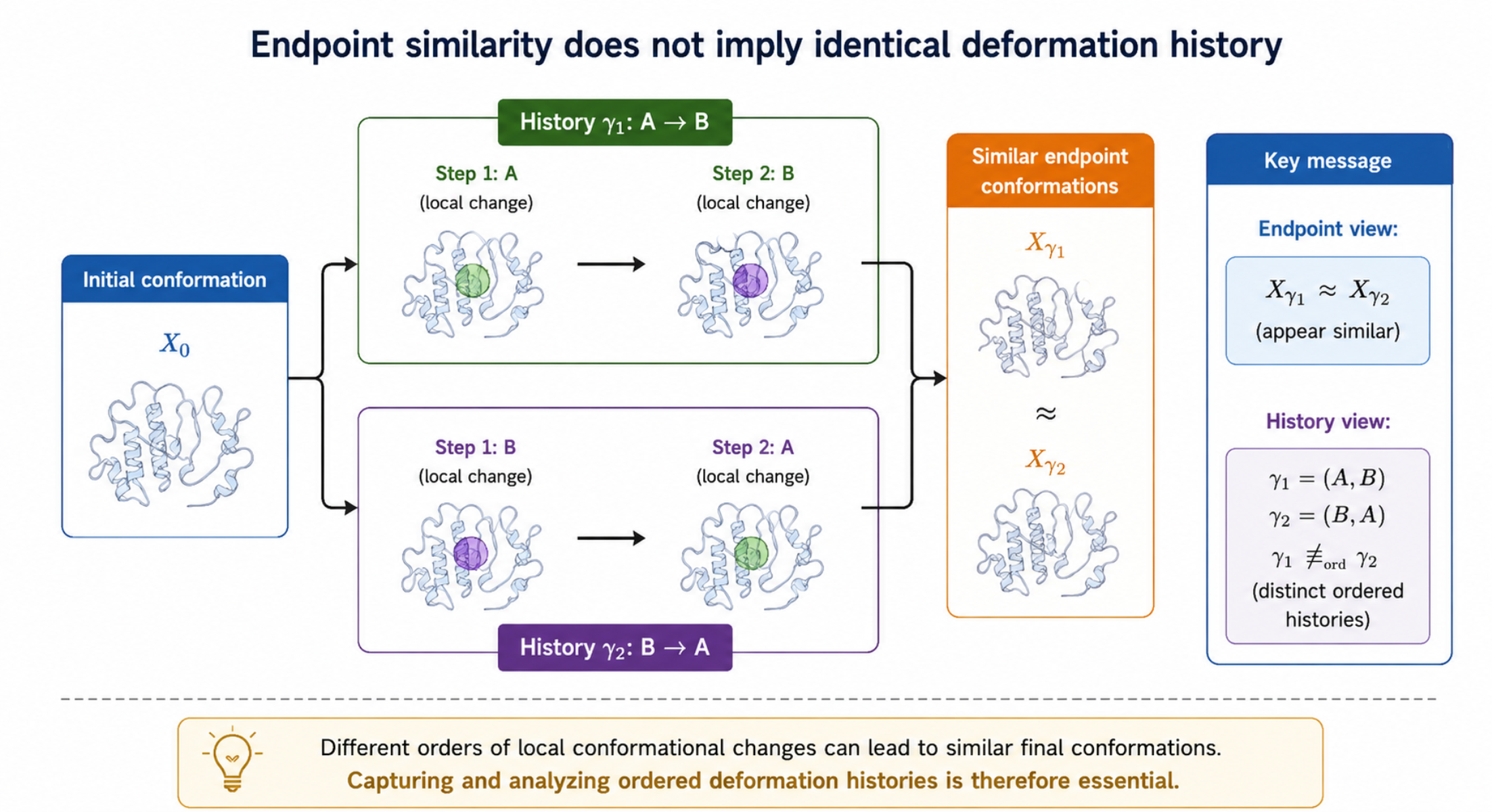}
    \caption{
    Endpoint similarity does not imply identical deformation history.
    Starting from the same initial conformation \(X_0\), two localized
    conformational changes \(A\) and \(B\) may occur in different orders,
    defining the deformation histories
    \(\gamma_1=(A,B)\) and \(\gamma_2=(B,A)\).
    Although their endpoint conformations may be similar,
    \(X_{\gamma_1} \approx X_{\gamma_2}\), the ordered histories remain
    distinct, \(\gamma_1 \not\equiv_{\mathrm{ord}} \gamma_2\).
    }
    \label{fig:endpoint-history}
\end{figure}

Figure~\ref{fig:endpoint-history}. Starting from the same initial conformation, two localized conformational changes may occur in different orders, producing distinct deformation histories while remaining similar at the level of endpoint geometry. This observation motivates the central question addressed in this work: how can a representation preserve deformation history rather than only the final structure?

Current computational methods are highly effective at representing, predicting, and generating protein states. Structure-prediction and protein-design approaches, including AlphaFold \citep{Jumper2021}, AlphaFold~3 \citep{Abramson2024}, RFdiffusion \citep{Watson2023}, and ProteinMPNN \citep{Dauparas2022}, have transformed biomolecular structure inference and design. Meanwhile, molecular-dynamics, coarse-grained, and recent geometric learning methods provide access to trajectory-level descriptions, including learned time-coarsened dynamics \citep{Klein2023}, multiresolution stochastic surrogates \citep{Schreiner2023}, geometric trajectory diffusion \citep{Han2024}, and generative modeling of molecular-dynamics trajectories \citep{Jing2024}. Nevertheless, once a deformation history is reduced to an endpoint conformation, the order in which local rearrangements occurred may be lost.

This loss can be biologically significant. A localized pitch deformation followed by bending may produce a different geometric or functional outcome than the same two perturbations applied in the reverse order \citep{Changeux2005,Motlagh2014,HenzlerWildman2007}. Such order dependence may be relevant to allosteric regulation, conformational switching, mutation-driven remodeling, and epistasis. This raises a central question: how can protein deformation be represented in a way that preserves ordered transport history without reducing distinct deformation pathways to their endpoint geometry?

\subsection*{Beyond endpoint representations}
\vspace{-0.3em}

Many protein-analysis workflows compare conformations through endpoint descriptors, including root-mean-square deviation (RMSD), distance maps, contact patterns, torsion summaries, latent embeddings, or generated structural states \citep{Jing2021}. These representations are essential, but they may identify two conformations as similar even when the ordered sequence of local rearrangements that produced them is different \citep{moradi2024review}. 

The present framework is intended to provide a mathematical basis for future computable descriptors of path-dependent protein response. Structural ensembles, NMR models, molecular-dynamics trajectories, or samples from generative protein models \citep{Yim2023,Huguet2024} could eventually be converted into discrete quaternionic transport fields, from which order-memory and spectral-response signatures could be extracted and compared across deformation histories.

To address this problem, we develop a deformation-first geometric framework based on the transport of adapted frames along the protein backbone. Local frame rotations are represented by quaternionic variables, providing a compact and nonsingular description of three-dimensional rotational deformation. Ordered deformation paths are then treated as primary objects, since two local perturbations that do not commute produce distinct transport histories when applied in different orders, even when their endpoint descriptors are close.

These histories generate a noncommutative ordered transport algebra
\[
\mathcal A_{\mathrm{tr}}(\Sigma),
\]
defined on the physical quaternionic locus
\(\Sigma\subset P_{\mathbb C}\). Its noncommutativity records the
order of admissible deformation paths, while a commutative shadow
separates order-sensitive transport memory from collapsed response
observables. From this transport structure, we construct a
Dirac-type spectral layer consisting of a global operator
\(\mathfrak D\), intrinsic local spectral germs
\(\mathfrak L_\xi(\mathfrak D)\), a renormalized spectral density, and
a mixed response form. The response geometry is therefore extracted
from ordered transport rather than postulated on a predetermined space
of endpoint conformations.

The working hierarchy is
\begin{equation}
\begin{aligned}
P_{\mathbb C} \supset \Sigma
&\longrightarrow \mathcal A_{\mathrm{tr}}(\Sigma)
\longrightarrow \mathfrak D
\longrightarrow \bigl(\xi \mapsto \mathfrak L_\xi(\mathfrak D)\bigr)
\\
&\longrightarrow \rho^{\mathrm{ren}}_{\chi}(\xi,\bar{\xi})
\longrightarrow \mathcal G_{\chi}(\xi)
\longrightarrow \text{memory sectors}.
\end{aligned}
\label{eq:deformation-first-chain}
\end{equation}

For readers coming from protein science rather than noncommutative or
spectral geometry, the principal objects in
Eq.~\eqref{eq:deformation-first-chain} are summarized in
Table~\ref{tab:object-correspondence}.

\begin{table}[!htbp]
\centering
\caption{Correspondence between the mathematical objects of the
deformation-first hierarchy, their protein-science interpretation, and
their role in the construction.}
\label{tab:object-correspondence}
\begin{tabularx}{\textwidth}{
    >{\raggedright\arraybackslash}p{0.15\textwidth}
    >{\raggedright\arraybackslash}p{0.33\textwidth}
    >{\raggedright\arraybackslash}X}
\toprule
\textbf{Mathematical object}
&
\textbf{Protein-science interpretation}
&
\textbf{Role in the manuscript}
\\
\midrule

$P_{\mathbb C}$
&
Complex deformation space containing admissible deformation modes after
complexification.
&
Provides the ambient space in which deformation families and mixed
spectral-response geometry are formulated.
\\
\cmidrule(lr){1-3}

$\Sigma$
&
Physical quaternionic locus compatible with admissible
protein-derived transport states.
&
Provides the physical support on which ordered transport paths are
defined.
\\
\cmidrule(lr){1-3}

$q(\ell)$ and
$\Omega(\ell)$
&
Quaternionic lift of local backbone frames and the corresponding
infinitesimal rotational transport field.
&
Encodes local frame orientation, bending, twisting, and rotational
deformation along the backbone.
\\
\cmidrule(lr){1-3}

$\mathcal P_{\mathrm{ad}}(\Sigma)$ and
$\mathcal G_{\mathrm{ad}}(\Sigma)$
&
Admissible ordered deformation paths and their groupoid of
concatenations.
&
Makes deformation histories primary objects rather than endpoint
reductions.
\\
\cmidrule(lr){1-3}

$\mathcal A_{\mathrm{tr}}(\Sigma)$
&
Ordered transport algebra generated by admissible deformation histories.
&
Retains order-sensitive transport information; its noncommutativity
records that deformation $A$ followed by $B$ need not be equivalent to
$B$ followed by $A$.
\\
\cmidrule(lr){1-3}

$\sigma$
&
Unitary cocycle or geometric phase induced by the transport connection.
&
Provides the twisting datum associated with order-sensitive geometric
phase information.
\\
\cmidrule(lr){1-3}

$E:\mathcal A_{\mathrm{tr}}(\Sigma)\to C_c^\infty(\Sigma)$
&
Commutative shadow or endpoint-level projection.
&
Separates order-sensitive transport information from collapsed
commutative observables.
\\
\cmidrule(lr){1-3}

$\mathcal M_{\mathrm{ord}}=\ker E$
&
Order-memory sector.
&
Contains the transport information that is lost under the commutative
shadow.
\\
\cmidrule(lr){1-3}

$\mathfrak D$
&
Global Dirac-type operator associated with the infinitesimal transport
realization.
&
Provides the spectral realization from which local response quantities
are extracted.
\\
\cmidrule(lr){1-3}

$\mathfrak L_\xi(\mathfrak D)$
&
Intrinsic local spectral germ at the physical transport state
$\xi\in\Sigma$.
&
Localizes the spectral response while retaining intrinsic information
from the global Dirac datum.
\\
\cmidrule(lr){1-3}

$\rho^{\mathrm{ren}}_\chi$
&
Renormalized local spectral density.
&
Quantifies local spectral response after finite-part renormalization.
\\
\cmidrule(lr){1-3}

$\mathcal G_\chi
 = i\partial\bar\partial\rho^{\mathrm{ren}}_\chi$
&
Mixed response form.
&
Encodes local sensitivity of the deformation-response landscape.
\\
\cmidrule(lr){1-3}

$\Phi_{\chi,W}$
&
Local response potential, whenever local
$\partial\bar\partial$-exactness holds.
&
Provides a scalar potential representation of the response geometry on
suitable open sets.
\\
\bottomrule
\end{tabularx}
\end{table}

Equation~\eqref{eq:deformation-first-chain}, together with 
Table~\ref{tab:object-correspondence}, summarizes how visible protein deformation modes are lifted into ordered transport histories and then converted into spectral-response quantities. 

\subsection*{Main contributions}
\vspace{-0.3em}

The main contributions of this work are as follows. First, we formulate ordered protein deformations as primary transport objects, while endpoint states are treated as reduced geometric descriptions. Second, we introduce a quaternionic transport algebra encoding noncommuting deformation histories. Third, we construct a Dirac-type spectral-response layer from the ordered transport structure. Fourth, we distinguish order-memory and response-memory sectors through their behavior under a commutative shadow. Fifth, we illustrate the framework on an idealized $\alpha$-helix with localized pitch and bending perturbations.

The remainder of the paper is organized as follows.
Section~\ref{sec:background} introduces the biological motivation and the mathematical background of ordered protein deformation.
Section~\ref{sec:constitutive-axioms} introduces the deformation-first framework. 
Section~\ref{sec:deformation-transport} introduces the deformation datum, admissible quaternionic transport paths, the ordered transport algebra, and its commutative shadow. 
Section~\ref{sec:canonical-realization} develops the canonical geometric realization and the normalized unitary cocycle. 
Section~\ref{sec:infinitesimal-dirac} constructs the infinitesimal transport module and the global Dirac operator. 
Section~\ref{sec:localization-response} defines intrinsic local spectral germs, renormalized response densities, and mixed response geometry. 
Section~\ref{sec:helical-realization} presents a minimal helical realization on an idealized \(\alpha\)-helix.
Section~\ref{sec:memory-collapse} separates order-memory and response-memory sectors.
Section~\ref{sec: limit} discusses limitations and empirical outlook, 
and Section~\ref{sec13} concludes the paper.

\subsection*{Reader's map}
\vspace{-0.3em}
Readers primarily interested in the protein-science motivation and computational interpretation may focus on Sections~\ref{sec:intro}--\ref{sec:background}, the conceptual hierarchy in Eq.~\eqref{eq:deformation-first-chain} and Table~\ref{tab:object-correspondence}, the illustrative \(\alpha\)-helical realization in Section~\ref{sec:helical-realization}, and the interpretation and outlook in Sections~\ref{sec:memory-collapse}--\ref{sec13}. Sections~\ref{sec:constitutive-axioms}--\ref{sec:localization-response} provide the formal noncommutative and spectral-geometric construction, including the transport groupoid and algebra, the geometric cocycle, the Dirac realization, and the intrinsic local spectral response. In particular, Section~\ref{sec:helical-realization} should be read as an endpoint-derived numerical illustration of order sensitivity, whereas the formal order-memory sector is defined algebraically through the commutative-shadow construction.

\section{Background and reader guide}\label{sec:background}
This section develops the biological and mathematical background for the deformation-first framework. We first explain why endpoint conformations alone may not retain the ordered history of protein deformation. We then review local backbone frames, quaternionic rotations, and the geometric origin of noncommutativity. After summarizing the mathematical domains used in the construction, we formulate the central problem by distinguishing endpoint-level representations from history-level representations. Figure~\ref{fig:order-aware-transport} provides a schematic overview of this conceptual organization.

\begin{figure*}[t]
    \centering
    \includegraphics[width=\textwidth]
    {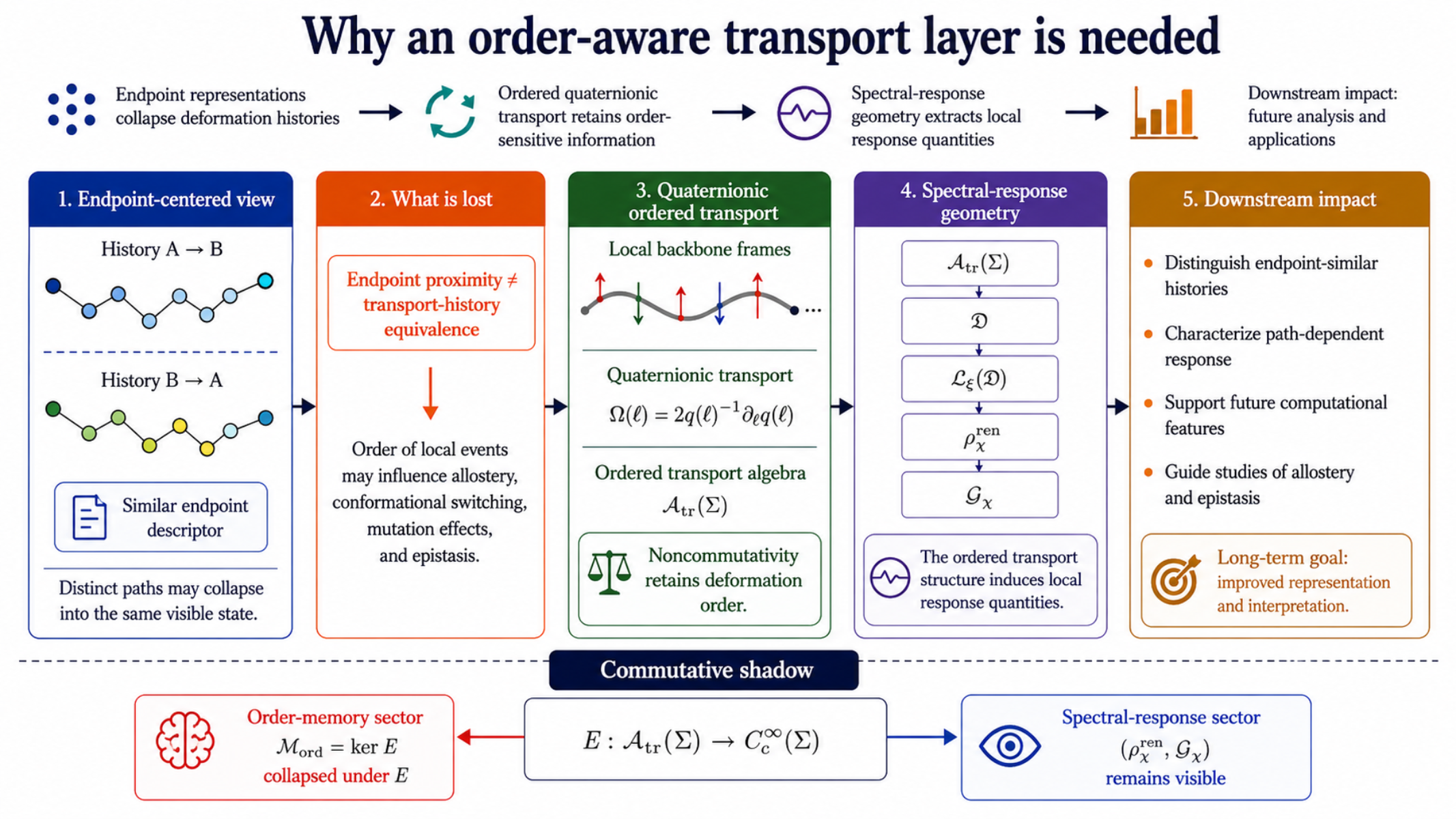}
    \caption{Conceptual overview of deformation-memory preservation. Endpoint-centered representations can collapse distinct ordered deformation histories into similar final conformations. The proposed quaternionic transport layer lifts local backbone frames to ordered transport variables, preserving information about the sequence of local rearrangements. The resulting noncommutative transport algebra separates order-memory information from endpoint-level response descriptors, while the spectral-response layer extracts local response quantities from the ordered transport geometry.}
    \label{fig:order-aware-transport}
\end{figure*}

\subsection{Protein deformation beyond endpoint conformations}
A protein conformation is often represented by a set of atomic coordinates, a backbone trace, or a reduced structural descriptor \citep{agnihotry2022protein}. Such representations are essential, but they mainly describe where the protein is observed after a perturbation. They do not necessarily retain the ordered sequence of local rearrangements that produced this state \citep{moradi2024review}.

This distinction is biologically relevant because protein function often depends on the organization of conformational changes. In allostery, a local perturbation may propagate through a sequence of structural couplings before influencing a distant functional site \citep{Changeux2005,Motlagh2014,HenzlerWildman2007}. In epistasis, the effect of one mutation may depend on whether another perturbation has already occurred \citep{StarrThornton2016}. Likewise, different intermediate rearrangements may guide a protein toward different functional basins even when their endpoint conformations are similar.

\subsection{Local backbone frames and quaternionic rotations}

To describe protein deformation geometrically, one may follow the
protein backbone and attach a local orthonormal frame along it. Such a
frame records the local orientation of the chain. Classical choices
include Frenet frames, while Bishop frames or related adapted frames are
useful near curvature degeneracies where Frenet torsion becomes
unstable \citep{Bishop1975}. Discrete framed-curve and elastic-rod
formulations provide corresponding tools for frame transport along
discretized backbones \citep{Bergou2008}.

Local frame rotations are conveniently represented through unit
quaternions, or equivalently through the Lie group
\(\mathrm{SU}(2)\), which double-covers
\(\mathrm{SO}(3)\). Let
\[
q(\ell)\in\mathrm{SU}(2)
\]
denote a quaternionic lift of the adapted frame, where \(\ell\) is arc
length. The associated infinitesimal transport field is

\begin{equation}
\Omega(\ell)=2q(\ell)^{-1}\partial_\ell q(\ell),
\label{eq:quaternionic-transport-field}
\end{equation}
which describes the instantaneous rotation of the local frame along the
backbone.

Quaternionic frame descriptions have proved useful in protein geometry,
while Bishop-frame and discrete-rod formulations provide robust
alternatives near framing singularities
\citep{Hanson1994,Bergou2008,Bishop1975}. Related frame-based
representations have also appeared in recent geometric learning models
for protein backbones \citep{Lin2023}.

\subsection{Why noncommutativity enters the problem}
Noncommutativity arises naturally from the kinematics of protein backbones. A local conformational change, such as a rotation about a backbone dihedral angle (i.e., a change in \(\phi\) or \(\psi\)), reorients the downstream backbone and thereby changes the local frame in which subsequent deformations are expressed. Consequently, applying two localized conformational changes in opposite orders generally yields different backbone geometries, even when the same elementary motions are involved.

This order dependence reflects an intrinsic property of protein deformation rather than a particular mathematical formalism. Representations based solely on endpoint conformations may therefore identify similar final structures while failing to distinguish the ordered deformation histories that generated them.

\subsection{Mathematical domains used in the construction}
The construction combines several mathematical domains.
First, differential geometry is used to describe adapted frames, local
transport, curvature, and connections along deformation spaces. Second,
quaternionic geometry is used to represent local rotations of protein
backbone frames through the double cover
\[
\mathrm{SU}(2)\to \mathrm{SO}(3),
\]
in connection with modern frame-based and equivariant geometric
representations \citep{Fuchs2020,Yim2023,Ruhe2023}.
Third, noncommutative geometry provides an algebraic language for
ordered transport, where path concatenation does not collapse to a
commutative endpoint description
\citep{Connes1985}. Fourth, spectral geometry enters
through a Dirac-type operator and its local spectral response.

\subsection{Problem formalisation: representing deformation histories beyond endpoint descriptors}
\label{subsec:problem-formalisation}

A protein sequence is denoted by
\[
    s = (a_1,\ldots,a_L) \in \mathcal{A}^L,
\]
where \(\mathcal{A}\) is the amino-acid alphabet and \(L\) is the protein length. For a fixed sequence \(s\), let
\[
    \mathcal{X}_s
\]
denote the set of structural states associated with this sequence. An element
\[
    X \in \mathcal{X}_s
\]
may represent an atomic structure, a backbone trace, an NMR ensemble member, a molecular-dynamics frame, or a generated conformation.

Let \(\mathcal{Z}\) be a global representation space, with
\[
\mathcal{Z}_{\mathrm{end}} \subset \mathcal{Z},
\qquad
\mathcal{Z}_{\mathrm{hist}} \subset \mathcal{Z},
\]
denoting the endpoint- and history-level representation sectors,
respectively.
Endpoint-centered approaches use an endpoint representation map
\[
\phi_{\mathrm{end}}:\mathcal{X}_s\to \mathcal{Z}_{\mathrm{end}},
\]
together with an endpoint discrepancy
\[
\Delta_{\mathrm{end}}: \mathcal{Z}_{\mathrm{end}} \times \mathcal{Z}_{\mathrm{end}} \to \mathbb{R}_{\geq 0},
\]
where smaller values indicate closer endpoint representations.
This formulation includes many standard choices. For Cartesian coordinates, \(\phi_{\mathrm{end}}\) may be the aligned coordinate representation and \(\Delta_{\mathrm{end}}\) may be RMSD. For distance maps, \(\phi_{\mathrm{end}}\) may be the pairwise-distance matrix and \(\Delta_{\mathrm{end}}\) a Frobenius-type matrix distance. For contact maps, \(\phi_{\mathrm{end}}\) may be a binary contact representation and \(\Delta_{\mathrm{end}}\) a Hamming- or Jaccard-type distance. For torsion summaries, \(\phi_{\mathrm{end}}\) may encode backbone dihedral angles and \(\Delta_{\mathrm{end}}\) an angular discrepancy. For learned structural embeddings, \(\phi_{\mathrm{end}}\) may be a latent vector and \(\Delta_{\mathrm{end}}\) a Euclidean or cosine-based distance.

These choices are well suited to comparing final conformations. However, they share a structural limitation: their input is an endpoint state \(X\). If two deformation processes have already been collapsed to their final conformations, any information about the order of the local rearrangements that produced those conformations is no longer available to \(\phi_{\mathrm{end}}\) or \(\Delta_{\mathrm{end}}\).

To formalize this missing information, let \(\Gamma_s\) denote the space of deformation histories associated
with the protein sequence \(s\). An element of \(\Gamma_s\) is written as
\[
\gamma=(d_1,\ldots,d_n)\in\Gamma_s,
\]
where each \(d_i\) denotes a localized deformation event, such as a bend, twist, pitch change, mutation-induced rearrangement, ligand-induced response, or local backbone reorientation. The endpoint map
\[
    \Pi : \Gamma_s \to \mathcal{X}_s
\]
assigns to each history its final conformation,
\[
    X_\gamma = \Pi(\gamma).
\]
The central problem is that two histories may be close at the endpoint level while remaining different as ordered deformation processes. In other words, it may happen that
\[
    \Delta_{\mathrm{end}}\!\left(
        \phi_{\mathrm{end}}(X_{\gamma_1}),
        \phi_{\mathrm{end}}(X_{\gamma_2})
    \right) \approx 0,
\]
while
\[
    \gamma_1 \not\equiv_{\mathrm{ord}} \gamma_2.
\]
Here, \(\gamma_1 \not\equiv_{\mathrm{ord}} \gamma_2\) means that the two histories are not equivalent as ordered compositions of local deformation events. This distinction is relevant when the order of perturbations affects the protein response, as in allostery, conformational switching, mutation-order effects, or epistatic rearrangements.

The objective is to construct a history-level representation
\[
    \psi_{\mathrm{hist}} : \Gamma_s \to \mathcal{Z}_{\mathrm{hist}},
\]
together with a history discrepancy
\[
    \Delta_{\mathrm{hist}} :
    \mathcal{Z}_{\mathrm{hist}} \times \mathcal{Z}_{\mathrm{hist}}
    \to \mathbb{R}_{\geq 0},
\]
such that histories collapsed by endpoint descriptors may still be separated at the history level:
\[
    \Delta_{\mathrm{end}}\!\left(
        \phi_{\mathrm{end}}(X_{\gamma_1}),
        \phi_{\mathrm{end}}(X_{\gamma_2})
    \right) \approx 0,
\]
but
\[
    \Delta_{\mathrm{hist}}\!\left(
        \psi_{\mathrm{hist}}(\gamma_1),
        \psi_{\mathrm{hist}}(\gamma_2)
    \right) > 0.
\]
This formulation also explains why the history cannot be treated as a purely symbolic list of events. First, protein deformation is spatial: local events act on a three-dimensional backbone, and their effect depends on the local frame in which they occur. Second, two histories with the same symbolic labels may produce different geometric effects if the local backbone orientation has changed. Third, direct comparison of event lists does not capture the accumulated rotational transport induced along the backbone. A history representation must therefore encode not only the order of events, but also their local geometric action.

This is where quaternionic frame transport enters the construction. 
As described by Eq.~\eqref{eq:quaternionic-transport-field},
quaternionic frame transport provides a compact and nonsingular
representation of local three-dimensional backbone rotation. Its role
here is not merely to encode a symbolic sequence of events, but to
record how ordered local deformations act on the evolving backbone
frame.

For example, let \(A_\alpha\) denote a localized pitch perturbation and \(B_\beta\) a localized bending perturbation. We compare the two ordered histories
\[
    \gamma_{AB} = B_\beta \circ A_\alpha,
\]
and
\[
    \gamma_{BA} = A_\alpha \circ B_\beta,
\]
where \(B_\beta \circ A_\alpha\) means that \(A_\alpha\) is applied first and \(B_\beta\) second. These two histories contain the same local events, but in opposite orders. Their endpoint conformations may be close, yet their accumulated transport may differ.

At the transport level, the representation assigns to each history a transport object
\[
    \mathcal{T}(\gamma) \in \mathcal{Z}_{\mathrm{hist}}.
\]
Order dependence is then expressed by the possibility that
\[
    \mathcal{T}(\gamma_{AB}) \neq \mathcal{T}(\gamma_{BA}),
\]
even when their endpoint representations satisfy
\[
\Delta_{\mathrm{end}}
\left(
\phi_{\mathrm{end}}\bigl(\Pi(\gamma_{AB})\bigr),
\phi_{\mathrm{end}}\bigl(\Pi(\gamma_{BA})\bigr)
\right)
\approx 0.
\]

In the formal construction developed below, admissible histories are represented by paths
\[
    \gamma : [0,1] \to \Sigma,
\]
where \(\Sigma\) is the physical quaternionic locus inside the complex deformation space \(P_{\mathbb{C}}\). Their ordered compositions generate the ordered transport algebra
\[
\Atr(\Sig).
\]
The algebraic generators \(U_\gamma\) represent admissible transport histories, and multiplication in \(\Atr(\Sig)\) records their order of composition. The commutative shadow
\[
E:\Atr(\Sig)\to C_c^\infty(\Sig)
\]
then identifies the information that remains after collapse to endpoint-level observables. Its kernel
\[
\mathcal M_{\mathrm{ord}}:=\ker E
\]
is the order-memory sector, namely the part of the representation that is lost under commutative endpoint reduction.

The problem addressed in this work can therefore be summarized as
follows: construct a representation of protein deformation histories
that preserves local geometric transport order, distinguishes histories
that endpoint descriptors may collapse, and provides response quantities
derived from this ordered transport structure. Quaternionic frame
transport supplies the local rotational encoding,
\(\Atr(\Sig)\) preserves ordered composition,
\(E\) identifies the collapse to endpoint-level information, and the
spectral-response layer extracts local response geometry from the
resulting transport representation.

\section{Formal deformation-first framework}\label{sec:constitutive-axioms}
Having introduced the biological motivation and the mathematical tools used in the construction, we now formulate the deformation-first framework. The purpose of this section is to define the structural layers of the model: the ordered deformation geometry on the physical locus
\(\Sig\subset P_{\C}\),
its infinitesimal and spectral realization, and the resulting local response theory. The central point is that the transport layer is constructed before the endpoint response layer, so that order-sensitive deformation history is not lost at the beginning of the formalization.

\begin{definition}[Connection-grounded deformation model]\label{def:ordered-deformation-model}
A \emph{connection-grounded deformation model} is a layered package
\[
\mathfrak S=(\mathfrak S_{\mathrm{def}},\mathfrak S_{\mathrm{inf}},\mathfrak S_{\mathrm{loc}})
\]
consisting of the following data.
\begin{enumerate}
\item The \emph{deformation-and-order layer}
\[
\mathfrak S_{\mathrm{def}}:=(P_{\C},\Sig,\mathcal Q,\mathscr L,\nabla^\mathscr L,\Pad(\Sig),\Gad(\Sig),\sigma,\Atr(\Sig),E);
\]

\item The \emph{infinitesimal-and-spectral layer}
\[
\mathfrak S_{\mathrm{inf}}:=(\EP,g_{\EP},\Hh,c_{\Hh},\pi,\Atr^\infty(\Sig),d_{\EP},\mathfrak D),
\]
providing the smooth infinitesimal realization of
\(\mathfrak S_{\mathrm{def}}\);

\item The \emph{local-response layer}
\[
\mathfrak S_{\mathrm{loc}}:=\Loc,
\]
together with the local spectral objects.
\end{enumerate}
\end{definition}
The three layers are linked by the deformation-first chain introduced
in Eq.~\eqref{eq:deformation-first-chain}, in which the
infinitesimal-and-spectral layer and the local-response layer are
constructed successively from the deformation-and-order layer.

\section{Connection-grounded complex deformation space and ordered transport}\label{sec:deformation-transport}
\subsection{Connection-grounded deformation geometry}

We now formalize the geometric framework underlying the transport-based deformation model. The basic geometric structure is encoded in the following connection-grounded deformation datum.

\begin{definition}[Connection-grounded deformation datum]\label{def:deformation-space}
A \emph{connection-grounded deformation datum} is a quintuple
\[
(P_{\C},\Sig,\mathcal Q,\mathscr L,\nabla^\mathscr L)
\]
satisfying the following conditions.
\begin{enumerate}
\item \(P_{\C}\) is a complex manifold.

\item \(\Sigma\subset P_\C\) is a connected smooth totally real embedded submanifold representing admissible physical transport states, called the \emph{physical quaternionic locus}.

\item \(\mathcal Q=\{(U_\alpha,\kappa_\alpha)\}_{\alpha\in A}\) is a quaternion-compatible atlas on \(\Sig\) locally modeled on admissible \(\mathfrak{su}(2)\)-connection fields and admitting holomorphic extensions to the corresponding complexified \(\mathfrak{sl}(2,\C)\)-connection fields on \(P_{\C}\).

\item \(\mathscr L\to\Sigma\) is a Hermitian line bundle endowed with a unitary connection \(\nabla^\mathscr L\), whose curvature
\[
F_{\nabla^\mathscr L}=i\,\Theta,
\qquad
\Theta\in \Omega^2(\Sig;\R)
\]
provides the first order-sensitive twisting datum.
\end{enumerate}

\end{definition}

\subsection{Ordered transport structures}

Ordered deformation histories are treated as primary geometric
objects rather than endpoint reductions. The resulting transport
structure is organized by admissible ordered paths and their
concatenations.

\begin{definition}[Admissible transport groupoid]\label{def:transport-groupoid}
An \emph{admissible ordered quaternionic path} in \(\Sig\) is an equivalence class \([\gamma]\) of piecewise \(C^1\) maps
\[
\gamma:[0,1]\to \Sig.
\]
Let \(\Pad(\Sig)\) be a chosen class of such paths, stable under constant paths, reversal, and ordered concatenation. The associated \emph{admissible transport groupoid}
\[
\Gad(\Sig)\rightrightarrows \Sig
\]
is the small groupoid defined as follows:
\begin{itemize}
\item its object space is \(\Sig\);
\item its arrow space is \(\Pad(\Sig)\);
\item for \(\gamma\in\Pad(\Sig)\), the source and target maps are
\[
 s(\gamma):=\gamma(0),\qquad t(\gamma):=\gamma(1);
\]
\item the unit at \(\xi\in\Sig\) is the constant class \(\id_\xi\);
\item the inverse of \(\gamma\) is the reversal class \(\gamma^{-1}\);
\item if \(t(\gamma_1)=s(\gamma_2)\), the product \(\gamma_2\circ\gamma_1\) is the ordered concatenation.
\end{itemize}
\end{definition}

\begin{definition}[Ordered transport algebra]\label{def:transport-algebra}
Fix a normalized unitary cocycle
\[
\sigma:\Gad(\Sig)^{(2)}\to U(1)
\]
on the admissible transport groupoid, and let
\[
\Gad(\Sig)^{(2)}_{\mathrm{ret}}
:=
\{(\gamma_2,\gamma_1)\in \Gad(\Sig)^{(2)}:\exists \xi\in \Sig\text{ such that } \gamma_2\circ\gamma_1=\id_\xi\}
\]
denote the set of composable pairs whose product is a unit.
An \emph{ordered transport algebra}
associated with \((\Gad(\Sig),\sigma)\) is an involutive algebra
\[
\Atr(\Sig)
\]
containing \(C_c^\infty(\Sig)\) as a distinguished nondegenerate commutative\(*\)-subalgebra, together with transport generators
\[
U_\gamma,\qquad \gamma\in \Gad(\Sig)\setminus \Sig.
\]
\end{definition}

\begin{definition}[Commutative shadow]\label{def:commutative-shadow}
Let \(\Atr(\Sig)\) be an ordered transport algebra in the sense of Definition~\ref{def:transport-algebra}. A \emph{commutative shadow} of \(\Atr(\Sig)\) is a linear map
\[
E:\Atr(\Sig)\to C_c^\infty(\Sig)
\]
which is a \(C_c^\infty(\Sig)\)-bimodule projection onto the coefficient algebra.
Its kernel
\[
\mathcal M_{\mathrm{ord}}:=\ker E
\]
is called the \emph{order-memory sector}. Equivalently,
\[
\Atr(\Sig)=C_c^\infty(\Sig)\oplus \mathcal M_{\mathrm{ord}}
\]
as \(C_c^\infty(\Sig)\)-bimodules.
\end{definition}

The full algebra retains order-sensitive transport memory, whereas the commutative shadow removes the purely ordered transport contribution while preserving collapsed response observables. This prepares the commutative collapse established later in the paper.

\section{Canonical geometric prototype and cocycle input for the smooth realization}
\label{sec:canonical-realization}

We now pass from the abstract deformation datum of Section~\ref{sec:deformation-transport} to a canonical smooth geometric realization on the physical locus \(\Sigma \subset P_{\mathbb C}\).
The purpose of the present section is twofold. First, we fix the canonical smooth geometric setting underlying the later transport construction. Second, we construct the canonical geometric cocycle that will serve as the twisting datum for the smooth transport realization developed in Section~\ref{sec:infinitesimal-dirac}.

\begin{definition}[Canonical geometric prototype]\label{def:canonical-prototype}
Let
\[
(P_{\C},\Sig,\mathcal Q, \mathscr L,\nabla^\mathscr L)
\]
be a connection-grounded deformation datum in the sense of Definition~\ref{def:deformation-space}, and let \(\Theta\in \Omega^2(\Sig;\mathbb R)\) be the curvature form determined by
\[
F_{\nabla^\mathscr L}= i\,\Theta.
\]
A \emph{canonical geometric prototype} consists of the above datum together with the following assumptions:
\begin{enumerate}
\item $\Sigma$ is connected, oriented, complete, and endowed with a real-analytic Riemannian metric $g$;
\item the curvature form $\Theta$ is closed and satisfies
\[
\frac{1}{2\pi}[\Theta]\in H^2(\Sig;\mathbb Z).
\]
\end{enumerate}
\end{definition}

Fix the Hermitian line bundle with unitary connection
$(\mathscr L,\nabla^\mathscr L)$
from the deformation datum, together with a good geodesic cover of \(\Sigma\) by relatively compact geodesically convex charts. The curvature form \(\Theta\) provides the canonical order-sensitive phase input, encoded through the holonomy of \(\nabla^\mathscr L\) along the physical transport core.

\begin{definition}[Canonical admissible paths]\label{def:canonical-paths}
Let $(\Sigma,g)$ be the physical transport core of a canonical geometric prototype and let
\[
\mathcal U=\{U_\alpha\}_{\alpha\in A}
\]
be the fixed good geodesic cover.
We denote by
\[
\Pad(\Sig)
\]
the class of admissible ordered quaternionic transport paths admitting piecewise geodesic representatives subordinate to the cover $\mathcal U$.
Its elements are called \emph{canonical admissible paths}.
This class is stable under constant paths, reversal, and ordered concatenation.
\end{definition}

The piecewise geodesic condition serves only to produce a canonical smooth realization of the ordered transport structure on the physical quaternionic locus.

\begin{definition}[Local geometric cocycle]\label{def:geometric-cocycle}
Let
\[
(\gamma_2,\gamma_1)\in \Gad(\Sig)^{(2)}
\]
be a composable admissible pair contained in a geodesically convex chart
\(
U\subset\Sigma.
\)
We say that $(\gamma_2,\gamma_1)$ is \emph{$U$-elementary} if there exist points $\xi,\eta,\zeta\in U$ such that $\gamma_1$ and $\gamma_2$ admit representatives given by the geodesic segments
\[
\gamma_{\xi\eta}:\xi\to\eta,
\qquad
\gamma_{\eta\zeta}:\eta\to\zeta
\]
contained in $U$.
Let
\[
\Delta_U(\xi,\eta,\zeta)\subset U
\]
denote the oriented geodesic $2$-simplex with ordered vertices $(\xi,\eta,\zeta)$.
The associated \emph{local geometric cocycle} is defined by
\[
\sigma_U(\gamma_2,\gamma_1)
:=
\exp\!\Bigl(i\int_{\Delta_U(\xi,\eta,\zeta)}\Theta\Bigr).
\]
\end{definition}

The integrality of $\Theta$ ensures that these local phase factors patch coherently under subdivision and triangulation. The cocycle identity is the holonomy identity obtained by applying Stokes' theorem to geodesic tetrahedra.

\begin{proposition}[Global geometric cocycle]\label{prop:global-geometric-cocycle}
For the canonical geometric prototype, the local geometric cocycle $\sigma_U$ extends to a normalized unitary cocycle
\[
\sigma:\Gad(\Sig)^{(2)}\to U(1)
\]
on the admissible transport groupoid, satisfying
\[
\sigma(\gamma_2,\gamma_1)=\sigma_U(\gamma_2,\gamma_1)
\]
for every $U$-elementary composable pair subordinate to a geodesically convex chart $U$ of the fixed good cover.
\end{proposition}

\begin{remark}
The present section constructs the geometric phase cocycle
\[
\sigma:\Gad(\Sig)^{(2)}\to U(1),
\]
but not yet the smooth ordered transport algebra.
The latter requires localization from individual admissible transport classes to \emph{localized smooth families}, realized in the present framework by admissible local bisections. Consequently, the cocycle constructed here serves as the twisting datum for the transport algebra introduced in Section~\ref{sec:infinitesimal-dirac}.
\end{remark}
\section{Infinitesimal transport and the global Dirac operator}\label{sec:infinitesimal-dirac}

The geometric cocycle constructed in Section~\ref{sec:canonical-realization} now serves as the twisting datum for the smooth realization of the deformation-first transport chain. The purpose of the present section is to pass from admissible ordered transport to its infinitesimal, metric-Clifford, and Dirac realizations on the physical locus \(\Sigma.\)

\subsection{Canonical infinitesimal module}

\begin{definition}[Canonical infinitesimal module]\label{def:canonical-EP}
Let
\[
(P_{\mathbb C},\Sigma,\mathcal Q,\mathscr L,\nabla^\mathscr L,g,\Theta)
\]
be a canonical geometric prototype.
The associated \emph{canonical infinitesimal module} is
\[
\mathcal E_P:=\Gamma(T\Sigma).
\]
Its anchor map is the identity
\[
\rho_{\EP}:\EP\to \Gamma(T\Sig),
\qquad
\rho_{\EP}(X):=X.
\]
For each $\xi\in \Sig$, one has
\[
\{\rho_{\EP}(X)(\xi):X\in \EP\}=T_\xi\Sig,
\]
so \(\EP\) realizes the full space of admissible infinitesimal directions. The Riemannian metric \(g\) induces a positive-definite symmetric \(C^\infty(\Sigma)\)-bilinear form on \(\EP\).
\end{definition}

\subsection{Canonical metric-Clifford realization}
Fix a chosen spin$^c$ structure on $\Sigma$, with spinor bundle
\[
S\to \Sig.
\]
The Hermitian line bundle $\mathscr L\to \Sig$ already fixed in the connection-grounded deformation datum serves as the twisting bundle, and $|\Lambda|^{1/2}$ denotes the half-density bundle on $\Sig$. Define
\[
\mathcal V:=S\otimes \mathscr L\otimes |\Lambda|^{1/2}.
\]

\begin{definition}[Canonical metric-Clifford datum]\label{def:canonical-clifford}
Let
\[
\EP=\Gamma(T\Sig)
\]
be the canonical infinitesimal module of Definition~\ref{def:canonical-EP}.
The Riemannian metric $g$ induces the positive-definite symmetric
$C^\infty(\Sig)$-bilinear form
\[
g_{\EP}:\EP\times\EP\to C^\infty(\Sig),
\qquad
g_{\EP}(X,Y)(\xi):=g_\xi(X(\xi),Y(\xi)).
\]
Define the Hilbert space
\[
\Hh:=\mathscr L^2(\Sig,\mathcal V).
\]
Let
\[
\operatorname{Cl}(T\Sig,g)
\]
be the Clifford bundle associated with $(T\Sig,g)$. Fiberwise Clifford multiplication on $S$, extended trivially on $\mathscr L\otimes |\Lambda|^{1/2}$, induces the bounded multiplication representation
\[
c_{\Hh}:\Gamma_c\bigl(\operatorname{Cl}(T\Sig,g)\bigr)\to \mathcal B(\Hh),
\qquad
c_{\Hh}(a)\psi:=a\cdot \psi.
\]
The quadruple
\[
(\EP,g_{\EP},\Hh,c_{\Hh})
\]
is called the canonical metric-Clifford datum.
\end{definition}

This realizes the infinitesimal transport structure as a canonical metric-Clifford datum on the Hilbert space
\(
\Hh.
\)

\subsection{Canonical smooth transport representation}
\begin{definition}[Admissible local bisection]\label{def:admissible-local-bisection}
An \emph{admissible local bisection} on $\Sig$ is a pair
\[
b=(U,\beta),
\]
where $U\subset \Sig$ is open and
\[
\beta:U\to \beta(U)\subset \Sig
\]
is a smooth diffeomorphism onto its image, such that every \(x\in U\) admits an open neighborhood
\[
U_x\subset U
\]
and an element $V_x$ of the fixed good geodesic cover of $\Sig$ satisfying
\[
y,\beta(y)\in V_x,
\qquad y\in U_x .
\]
and the unique geodesic segment in $V_x$ joining $y$ to $\beta(y)$ defines the canonical admissible transport class from $y$ to $\beta(y)$. Smooth geodesic connector data of local diffeomorphisms $\beta$ thus locally induce \[y\longmapsto b(y)\in \Gad(\Sig),\] with associated transport symbol \[U_b.\]
\end{definition}

\begin{definition}[Transport operator associated with an admissible local bisection]\label{def:transport-operator-b}
For each \(x\in U\), let
\[
\gamma_{b,x}:x\to \beta(x)
\]
denote the corresponding canonical admissible connector. Let
\[
\tau^{S\otimes \mathscr L}_{b,x}:(S\otimes \mathscr L)_x\longrightarrow (S\otimes \mathscr L)_{\beta(x)}
\]
be the parallel transport along a canonical representative of \(\gamma_{b,x}\) for the tensor-product connection on \(S\otimes \mathscr L\), and let
\[
(\beta_*^{1/2})_x:|\Lambda|^{1/2}_x\longrightarrow |\Lambda|^{1/2}_{\beta(x)}
\]
denote the canonical transport of half-densities induced by the diffeomorphism \(\beta\). Their tensor product defines a fiberwise unitary map
\[
\widetilde\tau_{b,x}:=\tau^{S\otimes \mathscr L}_{b,x}\otimes (\beta_*^{1/2})_x:
\mathcal V_x\longrightarrow \mathcal V_{\beta(x)}.
\]
For
\[
\psi\in C_c^\infty(\Sig,\mathcal V),
\]
the associated transport operator
\[
V_b:C_c^\infty(\Sig,\mathcal V)\to C_c^\infty(\Sig,\mathcal V)
\]
is defined by
\[ (V_b\psi)(y):= \begin{cases} \widetilde\tau_{b,\beta^{-1}(y)}\bigl(\psi(\beta^{-1}(y))\bigr), & y\in \beta(U),\\[4pt] 0, & y\notin \beta(U). \end{cases} \]
\end{definition}

\begin{definition}[Canonical smooth transport core]\label{def:canonical-smooth-core}
The \emph{canonical smooth transport core} is the complex vector space generated by the symbols
\[
U_b f,
\]
where \(b=(U,\beta)\) is an admissible local bisection and \(f\in C_c^\infty(U)\).
The restriction relations
\[
U_{(U,\beta)}f=U_{(U',\beta|_{U'})}f
\]
whenever \(U'\subset U\) is open and \(\operatorname{supp}(f)\subset U'\), and the identity relations
\[
U_{(U,\mathrm{id}_U)}f=\widetilde f,
\]
where \(\widetilde f\in C_c^\infty(\Sig)\) denotes the zero extension of \(f\).
We denote this quotient by
\[
\Atr^\infty(\Sig).
\]
\end{definition}

\begin{proposition}\label{prop:canonical-smooth-representation}
Let $M_{\widetilde f}$ be the corresponding multiplication operator on $\Hh$, and set on generators
\[
\pi_{\mathrm{can}}(U_b f):=V_b M_{\widetilde f}.
\]
Then this assignment descends to a well-defined nondegenerate $*$-representation
\[
\pi_{\mathrm{can}}:\Atr^\infty(\Sig)\to \mathcal B(\Hh).
\]
Moreover, each $V_b$ is a partial isometry on $\Hh$, and the dense subspace
\[
C_c^\infty(\Sig,\mathcal V)
\]
is invariant under
\[
\pi_{\mathrm{can}}\bigl(\Atr^\infty(\Sig)\bigr).
\]
\end{proposition}
\subsection{Canonical global Dirac realization}
\begin{definition}[Canonical global Dirac operator]\label{def:canonical-global-dirac}
Let
\[
\nabla^{\mathcal V}:C_c^\infty(\Sig,\mathcal V)\to C_c^\infty(\Sig,T^*\Sig\otimes \mathcal V)
\]
be the tensor-product connection on
\[
\mathcal V=S\otimes \mathscr L\otimes |\Lambda|^{1/2}.
\]
Let
\[
c_0:T^*\Sig\otimes \mathcal V\to \mathcal V
\]
denote fiberwise Clifford multiplication. The \emph{canonical global Dirac operator} is the densely defined first-order differential operator
\[
\mathfrak D_0:C_c^\infty(\Sig,\mathcal V)\subset \Hh\to \Hh,
\qquad
\mathfrak D_0\psi:=c_0\bigl(\nabla^{\mathcal V}\psi\bigr).
\]
Equivalently, if \((e_j)\) is a local \(g\)-orthonormal frame with dual coframe \((e^j)\), then
\[
\mathfrak D_0\psi=\sum_j c_0(e^j)\,\nabla^{\mathcal V}_{e_j}\psi.
\]
\end{definition}
Via the Riemannian metric $g$, we identify the differential of a function with an element of the
infinitesimal transport module by
\[
d_{\EP}f:=(df)^\sharp_g\in \EP.
\]
\begin{theorem}[Essential self-adjointness of the canonical Dirac operator]\label{thm:selfadjoint-dirac}
Let
\[
\mathfrak D_0:=c_0\circ \nabla^{\mathcal V}
\]
be the canonical global Dirac operator. Then
\(\mathfrak D_0\) is symmetric and essentially self-adjoint on
\[
C_c^\infty(\Sig,\mathcal V)\subset \Hh.
\]
Its closure
\[
\mathfrak D:\Dom(\mathfrak D)\subset \Hh\to \Hh,
\]
is the unique self-adjoint realization of the first-order Dirac-type differential operator.

Moreover, for every \(f\in C_c^\infty(\Sig),\) viewed through the identity bisection as an element of
\(\Atr^\infty(\Sig)\), the core \(C_c^\infty(\Sig,\mathcal V)\) is invariant under \(\pi_{\mathrm{can}}(f)=M_f\),
and one has on this core
\[
[\mathfrak D,\pi_{\mathrm{can}}(f)]\psi = c_0(df)\psi,
\qquad
\psi\in C_c^\infty(\Sig,\mathcal V).
\]
Equivalently, under the metric identification
\[
d_{\EP}f:=(df)^\sharp_g\in \EP,
\]
one has on \(C_c^\infty(\Sig,\mathcal V)\)
\[
[\mathfrak D,\pi_{\mathrm{can}}(f)] = c_{\Hh}(d_{\EP}f).
\]
In particular, the commutator with a coefficient element is represented by the bounded bundle endomorphism
of fiberwise Clifford multiplication by \(df\).
\end{theorem}

The infinitesimal, metric-Clifford, and Dirac stages of the deformation-first transport chain are therefore canonically realized on the physical locus \(\Sigma.\)

\section{Intrinsic local spectral response of the canonical Dirac datum}\label{sec:localization-response}

Section~\ref{sec:infinitesimal-dirac} produced the canonical global Dirac datum
\[
(\Atr^\infty(\Sig),\pi_{\mathrm{can}},\Hh,\mathfrak D),
\]
with underlying first-order differential expression
\[
\mathfrak D_0=c_0\circ \nabla^{\mathcal V}.
\]
The present section extracts from this global datum an intrinsic local spectral response structure on the physical locus \(\Sig.\)
Using the tangent-groupoid deformation of \(\mathfrak D_0,\)
we define the intrinsic local spectral germ \(\mathfrak L_\xi(\mathfrak D)\), together with its tangent representative \(D(\xi).\)
We then extend the resulting local spectral structure holomorphically to the complex deformation space and construct the associated renormalized local spectral density and mixed response form.

\subsection{Intrinsic local spectral germ}

Let
\[
\mathbb T\Sig
:=
T\Sig\times\{0\}
\sqcup
\Sig\times\Sig\times(0,1]
\]
be Connes' tangent groupoid of \(\Sig\), viewed as the standard smooth deformation from the pair groupoid \(\Sig\times\Sig\) to the tangent bundle \(T\Sig\).
The first-order differential expression \(\mathfrak D_0\) canonically determines an adiabatic family \(\mathfrak D^{\mathbb T}\) over \(\mathbb T\Sig\): for \(t>0\) the fiber operator is \(t\mathfrak D_0\), while at \(t=0\) over each tangent fiber \(T_\xi\Sig\) one obtains the normal operator of \(\mathfrak D_0\) at \(\xi\).

\begin{definition}[Intrinsic local spectral germ and canonical localization]\label{def:canonical-localization}
Let \(\mathfrak D^{\mathbb T}\) denote the tangent-groupoid deformation of \(\mathfrak D_0.\)
For each \(\xi\in\Sig\), the fiber of \(\mathfrak D^{\mathbb T}\) at \(t=0\) over the tangent space
\(T_\xi\Sig\) is a constant-coefficient Dirac-type operator
\[
\mathcal N_\xi(\mathfrak D_0):
C_c^\infty(T_\xi\Sig,\mathcal V_\xi)\to C_c^\infty(T_\xi\Sig,\mathcal V_\xi),
\]
called the \emph{normal operator} of \(\mathfrak D_0\) at \(\xi\).

The germ of \(\mathfrak D^{\mathbb T}\) along the tangent fiber \(T_\xi\Sig\times\{0\}\) is denoted by
\[
\mathfrak L_\xi(\mathfrak D)
\]
and is called the \emph{intrinsic local spectral germ} of the canonical Dirac datum at \(\xi\).

The associated \emph{canonical localized operator} is
\[
D(\xi):=\Loc_\xi(\mathfrak D):=\mathcal N_\xi(\mathfrak D_0),
\qquad \xi\in\Sig.
\]
\end{definition}

\begin{remark}[Normal operator versus full local germ]\label{rem:normal-vs-germ}
The tangent Dirac operator
\[
D(\xi)=\Loc_\xi(\mathfrak D)=\mathcal N_\xi(\mathfrak D_0),
\]
is only a representative of the intrinsic local germ
\[
\mathfrak L_\xi(\mathfrak D).
\]
While \(D(\xi)\) captures the first-order tangent model, the full germ retains the adiabatic local information of \(\mathfrak D^{\mathbb T}\). All local response quantities considered later are extracted from the full spectral germ and its induced heat-kernel data, not from the isolated constant-coefficient operator taken by itself.
\end{remark}

\subsection{Canonical invariance of the local spectral germ}

We now show that the intrinsic local spectral germ is independent of auxiliary choices, while its tangent representative is canonical up to unitary equivalence induced by the choice of normal coordinates and compatible local trivializations.

\begin{theorem}[Intrinsic well-definedness]\label{thm:local-germ-invariance}
For every \(\xi\in \Sig\), the intrinsic local spectral germ
\[
\mathfrak L_\xi(\mathfrak D)
\]
is well defined independently of all auxiliary choices. Any concrete tangent-fiber representative
\[
D(\xi)=\mathcal N_\xi(\mathfrak D_0)
\]
is determined canonically up to unitary equivalence on
\[
L^2(T_\xi \Sig,\mathcal V_\xi).
\]
Moreover, every such representative is symmetric and essentially self-adjoint on
\[
C_c^\infty(T_\xi \Sig,\mathcal V_\xi)\subset L^2(T_\xi \Sig,\mathcal V_\xi).
\]
In particular, the real localization stage of the deformation-first chain is canonically
realized by the family of unitary-equivalence classes
\[
\xi\longmapsto [D(\xi)],
\qquad \xi\in \Sig.
\]
\end{theorem}

\subsection{Complex analytic extension of the local germ}
The localization stage constructed above is intrinsically defined on the physical locus \(\Sig.\)
To pass from \(\Sig\) to the complex deformation space \(P_{\C}\), we impose a holomorphic extension hypothesis on the
intrinsic local spectral germ.

\begin{hypothesis}[Canonical holomorphic extension of the intrinsic local germ]\label{ax:canonical-analyticity}
There exists an open neighborhood
\[
\mathcal U\subset P_{\C}
\]
of \(\Sig\) such that the intrinsic local spectral germ extends to a
holomorphic family
\[
\zeta\longmapsto \mathfrak L_\zeta(\mathfrak D),
\qquad \zeta\in \mathcal U.
\]

Moreover, every point
\(\zeta_0\in \mathcal U\)
admits an open neighborhood
\(W\Subset \mathcal U\)
on which the germ can be represented by a holomorphic family of local
spectral operators
\[
D_W(\zeta).
\]
On
\(W\cap \Sig\),
these representatives belong to the canonical unitary-equivalence class
determined by Theorem~\ref{thm:local-germ-invariance}.

For each \(\zeta\in \mathcal U\) and each \(t>0\), the intrinsic local spectral germ admits a diagonal heat density
\[
k_t^{\mathrm{loc}}(\zeta,\bar\zeta),
\]
depending only on the germ itself and smooth in \((\zeta,\bar\zeta)\).
For every compact set \(K\Subset \mathcal U\), as \(t\downarrow0,\) the local heat density admits a short-time asymptotic expansion of the form
\[
k_t^{\mathrm{loc}}(\zeta,\bar\zeta)
\sim
\sum_{m\ge 0} a_m(\zeta,\bar\zeta)\, t^{(m-n)/2},
\qquad
n:=\dim \Sig,
\]
uniformly on \(K\).

After shrinking the parameter space if necessary, we continue to denote the resulting deformation space by \(P_{\C}.\)
\end{hypothesis}

\subsection{Local spectral density and mixed response form}

The local spectral response is extracted from the diagonal heat density associated with the intrinsic local spectral germ through a finite-part renormalization procedure.

\begin{definition}[Renormalized local spectral density and mixed response form]\label{def:heat-potential}
Fix normalization data consisting of a base point
\[
\xi_0\in \Sig,
\]
and a small-time cutoff
\[
\chi\in C_c^\infty([0,\infty))
\]
with \(\chi\equiv 1\) in a neighborhood of \(0\).

For \(\zeta\in P_{\C}\), let
\[
k_t^{\mathrm{loc}}(\zeta,\bar\zeta)
\]
be the diagonal heat density furnished by Hypothesis~\ref{ax:canonical-analyticity}. Since \(\chi\) has
compact support and the local heat expansion of Hypothesis~\ref{ax:canonical-analyticity} holds uniformly
on compact subsets, the truncated integral
\[
I_\varepsilon(\zeta,\bar\zeta)
:=
\int_\varepsilon^\infty
\chi(t)\,
\Bigl[
k_t^{\mathrm{loc}}(\zeta,\bar\zeta)
-
k_t^{\mathrm{loc}}(\xi_0,\bar\xi_0)
\Bigr]
\,\dd t
\]
admits an asymptotic expansion as \(\varepsilon\downarrow 0\). The \emph{renormalized local spectral
density} is the Hadamard finite part
\[
\rho_\chi^{\mathrm{ren}}(\zeta,\bar\zeta)
:=
\operatorname{FP}_{\varepsilon\downarrow 0} I_\varepsilon(\zeta,\bar\zeta).
\]
The associated \emph{mixed response form} is the smooth \((1,1)\)-form
\begin{equation}
\mathcal G_{\chi} := i\,\partial\bar\partial \rho_\chi^{\mathrm{ren}} .
\label{eq:mixed-response-form}
\end{equation}
Biologically, Eq.~\eqref{eq:mixed-response-form} describes how the renormalized spectral response changes across nearby deformation modes, thereby encoding local sensitivity of the protein response landscape.
In local holomorphic coordinates \((\zeta^a)\) on \(P_{\C}\), we write
\[
\mathcal G_{\chi}
=
i\,G^{(\chi)}_{a\bar b}(\zeta,\bar\zeta)\,
d\zeta^a\wedge d\bar\zeta^b.
\]
\end{definition}
The mixed response form depends only on the intrinsic local spectral germ and is independent of auxiliary local representative families.
\subsection{Local response potential}

The mixed response form is the primary geometric response object. A scalar potential is secondary and is
introduced only in open sets where the local \(\partial\bar\partial\)-equation can be solved.

\begin{definition}[Local response potential]\label{def:local-response-potential}
Let \(W\subset P_{\C}\) be an open set. A \emph{local response potential} for the canonical Dirac datum on
\(W\) is a real-valued smooth function
\[
\Phi_{\chi,W}\in C^{\infty}(W,\R)
\]
such that
\[
\mathcal G_{\chi}|_{W}=i\,\partial\bar\partial \Phi_{\chi,W}.
\]
If \(\xi_0\in W\), the potential may be normalized by the condition
\[
\Phi_{\chi,W}(\xi_0,\bar\xi_0)=0.
\]
Without a chosen base point in \(W\), it is understood only up to the addition of a pluriharmonic function.
\end{definition}

\begin{proposition}
Whenever a local response potential exists, any two such potentials differ by a pluriharmonic function.
\end{proposition}

\section{Minimal helical realization of ordered transport and response}
\label{sec:helical-realization}

We record a minimal protein-geometric realization of the abstract construction on an idealized \(\alpha\)-helix. The aim is not to model atomistic dynamics, but to exhibit, in the simplest nontrivial backbone geometry, the coexistence of a constant quaternionic transport baseline, order-sensitive geometric transport, a spectral-response density, and the commutative collapse of the ordered sector. The \(\alpha\)-helix is a natural test object for this purpose: its backbone geometry is regular, chiral, and classically described by nearly constant pitch and turn parameters \citep{Pauling1951,Kabsch1983}.

This example is deliberately minimal. It is not intended to reproduce atomistic molecular dynamics or to validate a quantitative protein model. Its purpose is to show, in familiar protein geometry, how the abstract formalism distinguishes endpoint similarity from ordered transport history. The \(\alpha\)-helix is useful because its undeformed geometry has nearly constant pitch, radius, curvature, and torsion. This provides a simple baseline from which localized pitch and bending perturbations can be compared in two different orders.

Accordingly, the numerical values of the endpoint-derived ordered-transport discrepancy \(\Delta_{\mathrm{ord}}\) and the spectral-response quantities reported below are illustrative outputs of this schematic realization and are not calibrated to any specific experimental protein, NMR ensemble, or molecular-dynamics trajectory.

Precisely because of this regularity, the helical setting provides a transparent bridge between the formal construction and protein geometry. It shows that the transport formalism already detects a biologically meaningful distinction in a familiar backbone geometry: two localized deformation histories may remain close at the level of endpoint conformations while producing distinct endpoint-derived ordered quaternionic transports. Thus, even in this deliberately minimal setting, the ordered transport calculation is sensitive to order-dependent geometric differences that remain small under conventional endpoint descriptors.

\subsection{Ideal helix and constant quaternionic generator}

Consider the residue-indexed \(C_\alpha\) backbone
\[
x_n=
\bigl(
R\cos(n\theta),\,
R\sin(n\theta),\,
nh
\bigr),
\qquad n=0,\ldots,N-1,
\]
where \(n\) indexes the residues and \(N\) denotes the number of residues in the helical segment. The standard ideal-helical values are
\[
\theta\simeq 100^\circ,
\qquad
h\simeq 1.5\,\text{\AA},
\qquad
R\simeq 2.3\,\text{\AA}.
\]
Equivalently, in arclength form one may write
\[
x(\ell)=
\left(
R\cos\frac{\ell}{a},\,
R\sin\frac{\ell}{a},\,
\frac{b \ell}{a}
\right),
\qquad
a=(R^2+b^2)^{1/2},
\qquad
b=\frac{p}{2\pi},
\]
where \(p\) is the pitch.  The curvature and torsion of this smooth helix are constant:
\[
\kappa=\frac{R}{R^2+b^2},
\qquad
\tau=\frac{b}{R^2+b^2}.
\]
Let \(F(\ell)=(t(\ell),n(\ell),b_F(\ell))\in \mathrm{SO}(3)\) be the Frenet frame along the ideal helix, and let \(q(\ell)\in \mathrm{SU}(2)\) be a spin lift of this frame. We use throughout the quaternionic transport field \(\Omega\) introduced in Eq.~\eqref{eq:quaternionic-transport-field}.

\begin{proposition}[Ideal helical transport]\label{prop:ideal-helical-transport}
For the ideal helix above, the quaternionic transport field is constant.  In a Frenet-frame convention, it has the form
\[
\Omega_0=\tau\,\mathbf i+\kappa\,\mathbf k,
\]
up to the harmless relabelling of the imaginary quaternionic basis induced by the initial frame choice. Consequently
\[
q(\ell)=q(0)\exp\left(\frac{\ell}{2}\Omega_0\right),
\qquad
U_0(L)=\exp\left(\frac{L}{2}\Omega_0\right),
\]
where \(U_0(L)\) denotes the ordered transport over a segment of arclength \(L\).
\end{proposition}

\begin{proof}
For a helix of constant curvature and torsion, the Frenet equations are
\[
t'=\kappa n,
\qquad
n'=-\kappa t+\tau b_F,
\qquad
b_F'=-\tau n.
\]
Thus, the angular velocity of the moving frame has constant components in the frame basis. Under the double cover \(\mathrm{SU}(2)\to \mathrm{SO}(3)\), this constant angular velocity lifts to a constant element of \(\mathfrak{su}(2)\). Using the convention of Eq.~\eqref{eq:quaternionic-transport-field}, the corresponding generator is \(\Omega_0=\tau\mathbf i+\kappa\mathbf k\), up to the base convention fixed at \(\ell=0\).  The exponential formula for \(q(\ell)\) follows by solving the constant-coefficient equation \(\partial_\ell q=q\Omega_0/2\).
\end{proof}

In the discrete residue-indexed realization, one chooses a smooth or discrete adapted frame \(F_n\) along the points \(x_n\), lifts it to \(q_n\in \mathrm{SU}(2)\), and defines the elementary transport increments
\[
g_n=q_n^{-1}q_{n+1}.
\]
For the ideal helix, these increments are approximately constant:
\[
g_n\simeq \exp\left(\frac{\Delta \ell}{2}\Omega_0\right),
\]
where \(\Delta \ell\) is the residue-to-residue arclength.  This gives a concrete biological sector of the ordered transport algebra \(\Atr(\Sigma)\).  In the undeformed baseline, the constant-generator structure makes the transport regular and nearly translation-invariant along the helical axis.

\subsection{Two noncommuting helical deformation histories}

The simplest way to reveal order sensitivity is to compare two localized deformation histories.  Let \(A_\alpha\) denote a localized pitch perturbation, for example,
\[
(r,\theta,z)\longmapsto (r,\theta,z+\alpha f(z)),
\]
where \(f\) is a smooth bump along the helix.  Let \(B_\beta\) denote a localized bending perturbation, for example
\[
x\longmapsto R_y(\beta g(z))x,
\]
where \(R_y\) is a rotation about a transverse axis and \(g\) is another smooth bump.  The two histories
\[
\gamma_{AB}=B_\beta\circ A_\alpha,
\qquad
\gamma_{BA}=A_\alpha\circ B_\beta
\]
need not define the same ordered transport, even when their visible endpoint descriptors are close.

For each history \(\gamma\in\{\gamma_{AB},\gamma_{BA}\}\), construct frames \(F_n^\gamma\), lifts \(q_n^\gamma\), increments
\[
g_n^\gamma=(q_n^\gamma)^{-1}q_{n+1}^\gamma,
\]
and partial ordered products
\[
U_\gamma(m)
=
g_0^\gamma g_1^\gamma\cdots g_{m-1}^\gamma.
\]

This product order follows the right-action convention
$q_{n+1}^\gamma=q_n^\gamma g_n^\gamma$ used in the numerical implementation.

The order-sensitive comparison is then
\begin{equation}
\delta_{\mathrm{ord}}(m)
=
\left\|
\log\left(U_{AB}(m)^{-1}U_{BA}(m)\right)
\right\|,
\qquad
\Delta_{\mathrm{ord}}=\delta_{\mathrm{ord}}(N-1).
\label{eq:ordered-memory-diagnostic}
\end{equation}

In the present helical realization, $\Delta_{\rm ord}$ measures
the discrepancy between the ordered spatial transports reconstructed from the
two resulting final backbones. Biologically, Eq.~\eqref{eq:ordered-memory-diagnostic} quantifies how the reconstructed backbone transport differs when the same localized perturbations are applied in different orders. In the current implementation it should therefore be interpreted as endpoint-derived order sensitivity rather than as an independently accumulated trajectory-history variable.

At the abstract transport level, one may use the group commutator
\begin{equation}
C_{AB}
=
U_AU_BU_A^{-1}U_B^{-1},
\label{eq:transport-commutator}
\end{equation}
which is trivial only when the two elementary transports commute in the realized sector.
Biologically, Eq.~\eqref{eq:transport-commutator} tests whether two localized deformation events are order-independent; a nontrivial commutator is precisely
order dependence of the elementary transports at the abstract transport level.

The numerical controls are consistent with this interpretation: when the two deformation histories are constrained to produce the same final backbone, the reconstructed transport discrepancy vanishes to numerical precision. This supports interpreting \(\Delta_{\mathrm{ord}}\) as an illustrative measure of endpoint-derived order sensitivity in the present realization.

\subsection{Spectral response and commutative shadow}

\begin{figure}[t]
\centering
\includegraphics[width=0.98\textwidth]{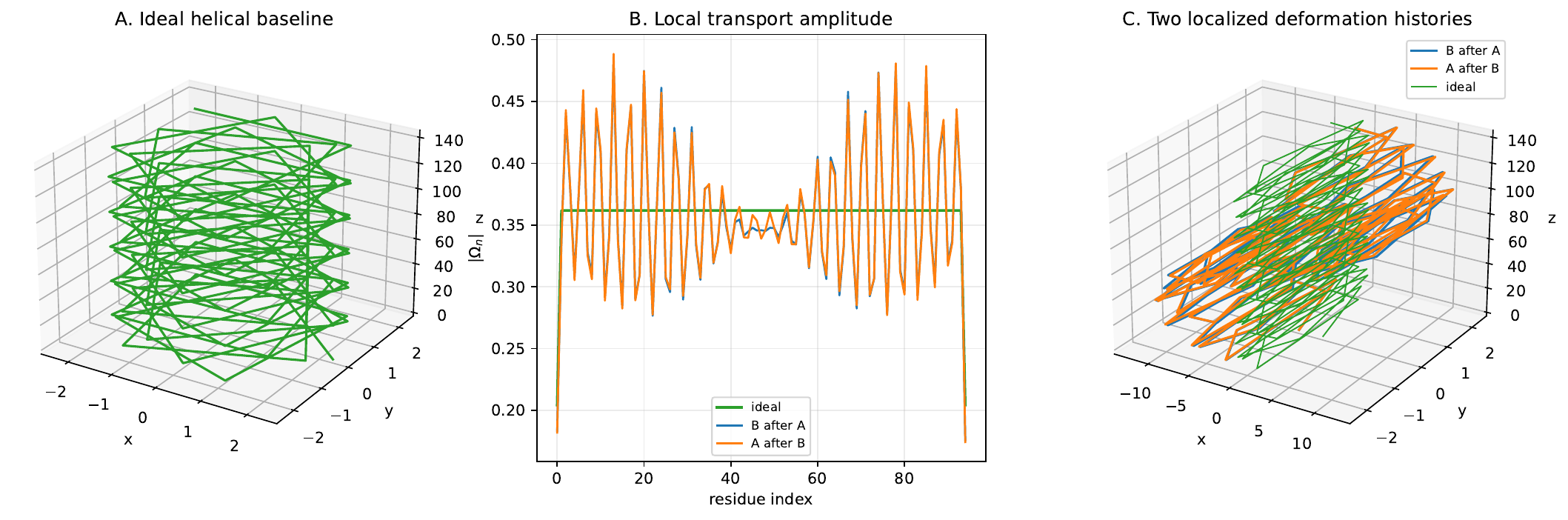}
\caption{Minimal helical realization of ordered quaternionic transport.
Panel A shows the ideal \(\alpha\)-helix baseline. 
Panel B shows the local quaternionic transport amplitude along the residue index. 
Panel C compares the final helical backbones obtained from the two
localized deformation histories.
The localized pitch deformation is denoted by \(A_{\alpha}\), and the localized bending deformation is denoted by \(B_{\beta}\). The two ordered histories are \(\gamma_{AB}=B_{\beta}\circ A_{\alpha}\) and \(\gamma_{BA}=A_{\alpha}\circ B_{\beta}\). With the standard convention for function composition, \(B_{\beta}\circ A_{\alpha}\) means that the pitch deformation is applied first and the bending deformation second, whereas \(A_{\alpha}\circ B_{\beta}\) means that the bending deformation is applied first and the pitch deformation second.}
\label{fig:helical-memory-geometry}
\end{figure}

The helical deformation family also carries the spectral-response layer.  Let
\(\mathcal P\subset\mathbb R^2\) be the pitch-bend parameter domain, and let
\[
\zeta_\gamma(\alpha,\beta)\in P_{\mathbb C},
\qquad \gamma\in\{\gamma_{AB},\gamma_{BA}\},
\]
denote the deformation mode obtained after frame extraction and
complexification.  The renormalized spectral density restricts to this sector
by
\[
\widehat{\rho}^{\,\mathrm{ren}}_{\chi,\gamma}(\alpha,\beta)
=
\rho_{\chi}^{\mathrm{ren}}
\bigl(
\zeta_\gamma(\alpha,\beta),
\overline{\zeta_\gamma(\alpha,\beta)}
\bigr),
\qquad
\widehat G_{\chi,\gamma}
=
i\,\partial\bar\partial
\widehat{\rho}^{\,\mathrm{ren}}_{\chi,\gamma}.
\]

\begin{figure}[t]
\centering
\includegraphics[width=0.98\textwidth]{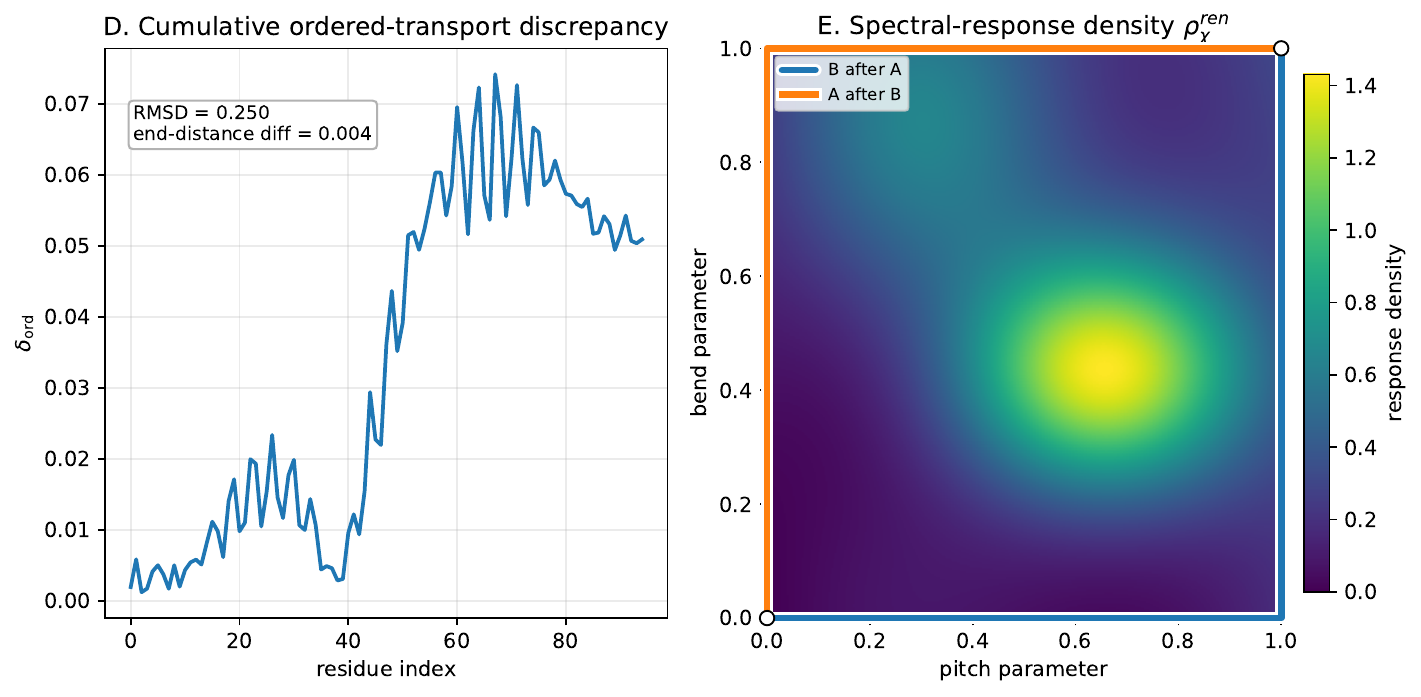}
\caption{Ordered-transport and spectral-response diagnostics.
Panel D shows the cumulative endpoint-derived ordered-transport
discrepancy. For the displayed parameters, the two final histories have a direct coordinate RMSD of 0.250 and an absolute difference in end-to-end distance of 0.004. 
Panel E shows the spectral-response density over the pitch-bend parameter domain. 
The key message is that the two
deformation orders may produce nearby final backbones while retaining a measurable difference in the spatial transports reconstructed from those backbones.}
\label{fig:helical-memory-response}
\vspace{-1.0em}
\end{figure}

Thus the same pitch-bend sector that detects transport order also carries a
local response landscape.

Finally, the commutative shadow
\[
E:\mathcal A_{\mathrm{tr}}(\Sigma)\to C_c^\infty(\Sigma)
\]
suppresses the formal order-memory sector $\mathcal M_{\mathrm{ord}}=\ker E$, while leaving the
response-memory data \(\widehat{\rho}^{\,\mathrm{ren}}_{\chi,\gamma}\) and \(\widehat G_{\chi,\gamma}\) available at the commutative level.  
This gives the helical realization of the memory separation stated in Section~\ref{sec:memory-collapse}.

This statement concerns the abstract commutative-shadow
construction. The numerical quantity $\Delta_{\rm ord}$ used above is
endpoint-derived and should not be identified directly with an independently
accumulated element of $\mathcal M_{\mathrm{ord}}$.

\begin{wrapfigure}{r}{0.5\textwidth}
\vspace{-1.0em}
\centering
\includegraphics[width=0.50\textwidth]{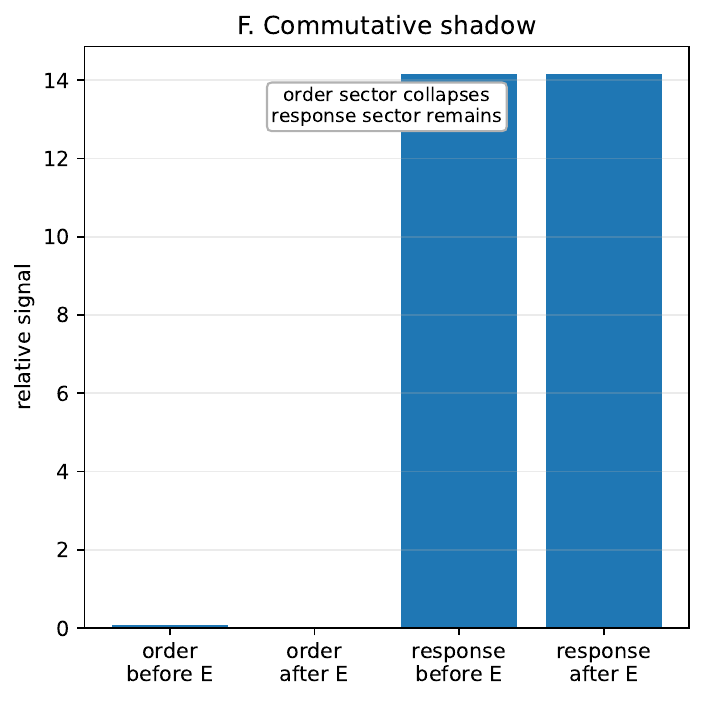}
\caption{Commutative shadow and memory-sector interpretation.
Panel F illustrates the commutative shadow: the order-sensitive sector collapses, whereas the response sector remains visible.}
\label{fig:helical-memory-shadow}
\vspace{-1.0em}
\end{wrapfigure}

Figures~\ref{fig:helical-memory-geometry}--\ref{fig:helical-memory-shadow} should therefore be read as a conceptual diagnostic. Figure~\ref{fig:helical-memory-geometry} shows the ideal helical baseline and two localized deformation histories. Figure~\ref{fig:helical-memory-response} shows that the cumulative endpoint-derived ordered-transport discrepancy can remain nonzero when endpoint-scale descriptors are close, and that the same pitch-bend sector supports a spectral-response landscape. Figure~\ref{fig:helical-memory-shadow} summarizes the separation between the order sector and the response sector: the order-sensitive contribution collapses under the commutative shadow, whereas the response sector remains visible. 
This diagnostic is intended to illustrate the formal definitions rather than to provide quantitative validation on a specific molecular system.
In particular, the numerical curves should be interpreted as a schematic realization of order-sensitive transport and response, not as a molecular-dynamics trajectory.

\subsection{Contribution to the response theory}
This helical calculation provides a minimal biological realization of the
deformation-first chain in Eq.~\eqref{eq:deformation-first-chain}.
The ideal helix supplies a constant-generator sector of \(\Atr(\Sig)\),
while localized pitch and bending perturbations produce two nearby but
noncommuting deformation histories. Their comparison yields a nonzero
endpoint-derived ordered-transport discrepancy
\(\Delta_{\mathrm{ord}}\), and the same two-parameter family carries a
pulled-back spectral-response density. Thus the example realizes, in one protein-like geometry, the separation
between order memory and response memory: the commutative shadow suppresses
the formal order-memory sector, whereas the response sector
determined by the Dirac realization remains.

\section{Two memory sectors and commutative collapse}\label{sec:memory-collapse}

The canonical realization separates two memory mechanisms, aligning with the view that protein function depends not only on static endpoints but also on ensembles, allostery, and dynamical response pathways \citep{HenzlerWildman2007}. 
Recent generative AI frameworks capture these non-equilibrium features through protein conformation generation models \citep{Lu2025} and force-guided \(\mathrm{SE}(3)\) diffusion models \citep{Wang2024}.

\begin{definition}[Memory sectors]\label{def:memory-sectors}

Fix normalization data
\[
(\chi,\xi_0).
\]
For the canonical realization, the \emph{response-memory sector} is the family of response data consisting of the intrinsic local spectral germs \(\mathfrak L_\xi(\mathfrak D)\) for \(\xi\in\Sigma\), the renormalized local spectral density \(\rho_\chi^{\mathrm{ren}}\), the mixed response form \(\mathcal G_\chi\), and, whenever \(\mathcal G_\chi\) is locally \(\partial\bar\partial\)-exact, the corresponding local response potentials \(\Phi_{\chi,W}\) on open sets \(W\subset P_{\C}\).

The \emph{order-memory sector} is the cocycle-twisted ordered transport sector
encoded by \(\sigma,\) and the \(C_c^\infty(\Sig)\)-subbimodule
\[
\mathcal M_{\mathrm{ord}}=\ker E.
\]
\end{definition}

\begin{proposition}[Collapse of the order-memory sector under the commutative shadow]\label{thm:commutative-collapse}
Let
\[
E:\Atr(\Sig)\to C_c^\infty(\Sig)
\]
be the commutative shadow map. Then
\[
E|_{\mathcal M_{\mathrm{ord}}}=0.
\]

Its image is the commutative coefficient algebra \(C_c^\infty(\Sig),\)
so the cocycle-twisted ordered transport structure has no nontrivial
commutative shadow.

By contrast, the response-memory sector determined by the intrinsic
local spectral germ,
the renormalized local spectral density,
and the mixed response form remains canonically determined by the Dirac realization. Whenever
\(\mathcal G_\chi\) is locally \(\partial\bar\partial\)-exact, the corresponding local response potentials \(\Phi_{\chi,W}\)
belong to the same surviving response sector.

Consequently, the commutative shadow suppresses the order-memory
sector while it does not in general suppress the spectral-response
sector.
\end{proposition}

\begin{proof}
From the decomposition
\[
\Atr(\Sig)
=
C_c^\infty(\Sig)
\oplus
\mathcal M_{\mathrm{ord}},
\]
one has
\[
E|_{\mathcal M_{\mathrm{ord}}}=0.
\]
Hence the cocycle-twisted ordered transport structure disappears under the commutative shadow.

The response-memory sector is determined by the intrinsic spectral
geometry associated with the canonical Dirac realization, namely by the local spectral germ and its induced response data from
Section~\ref{sec:localization-response}. It therefore survives the
commutative collapse of the order-sensitive sector
\[
\mathcal M_{\mathrm{ord}}.
\]
\end{proof}
This separation suggests that \(\mathrm{SE}(3)\)-equivariant representations may retain aspects of global geometric information even when local spatial descriptors collapse under commutative projections \citep{Fuchs2020,GarciaSatorras2021}.

\section{Limitations and outlook}\label{sec: limit}

The present work should be understood as a theoretical foundation. It
does not yet provide an atomistic molecular-dynamics simulation, an
experimentally validated predictor of protein response, or a complete
computational pipeline for empirical macromolecular systems. The
helical realization is intentionally minimal: it illustrates
endpoint-derived order sensitivity in nearby final backbones, rather than
providing quantitative validation of an independent trajectory-memory observable
on a specific protein system. The purpose of the present study is therefore to define the
mathematical objects carrying deformation order and response memory
before empirical estimators and validation protocols are developed.

A practical realization of the framework will require three main
steps. First, protein structures, NMR ensembles, or molecular-dynamics
trajectories must be converted into discrete local frames and
quaternionic transport increments along the backbone. Second,
order-sensitive quantities and numerical approximations of the
spectral objects
\[
D_{\mathrm{num}},\qquad
\mathfrak L^{\mathrm{num}}_\xi,\qquad
\rho^{\mathrm{ren}}_{\mathrm{num}},\qquad
\mathcal G_{\mathrm{num}}
\]
must be estimated and compared with standard endpoint descriptors.
Third, these signatures must be tested on systems where
path-dependence is expected, including allosteric transitions,
mutation-order effects, conformational switching, and epistatic
rearrangements.

At a practical level, molecular-dynamics trajectories could be converted frame by frame into residue-indexed adapted backbone frames and lifted to unit quaternions with a consistent choice of sign across neighboring residues and successive frames. The resulting discrete quaternionic increments could then be used to construct ordered transport descriptors along the backbone and, when temporal information is available, across successive trajectory frames. This trajectory-level construction is distinct from the present helical quantity \(\Delta_{\mathrm{ord}}\), which is reconstructed from final backbones and therefore measures endpoint-derived order sensitivity. For NMR ensembles, which generally do not provide an observed temporal ordering, a history-level analysis would additionally require experimentally supported or explicitly inferred transition information.

\section{Conclusion}\label{sec13}
In conclusion, this paper proposes a deformation-first framework for representing path-dependent response in protein structures. The central idea is that endpoint conformations and ordered deformation histories should be treated as distinct but coupled objects. By using quaternionic frame transport along the protein backbone, the framework provides a natural representation of local rotational deformation. By organizing admissible deformation paths into a noncommutative transport algebra, it records the fact that the order of local perturbations may matter.

The minimal helical realization illustrates order sensitivity in a familiar protein geometry. An idealized \(\alpha\)-helix provides a constant quaternionic transport baseline, while localized pitch and bending perturbations generate nearby final backbones that differ in the ordered quaternionic transport reconstructed from them. This yields a nonzero endpoint-derived ordered-transport discrepancy. At the same time, the spectral-response layer provides response objects that can remain visible after the order-sensitive sector collapses under the commutative shadow.

The resulting framework should therefore be viewed as a mathematical
foundation for deformation memory in protein dynamics. It does not
replace atomistic simulations or generative models, but provides a
formal layer for representing ordered histories that such models may
need to preserve. For protein researchers, the practical message is
that endpoint similarity and deformation-history equivalence are
distinct notions: two conformations may appear close under conventional
structural descriptors while retaining different ordered transport
histories. The framework provides a language for representing this
distinction and for designing future descriptors of path-dependent
protein response.

Future work will focus on empirical realization from molecular
structures and trajectories, numerical approximation of the
spectral-response objects, and integration with geometric deep learning
approaches for protein dynamics. More broadly, this deformation-memory layer could provide one component of future protein world models, understood as models that represent not only protein states, but also perturbations, ordered transitions, and path-dependent responses. In this sense, the proposed framework offers a theoretical bridge between protein geometry, deformation history, and future computational models of protein response.

\bmhead{Acknowledgements}
X.C is supported by a PhD grant from Région Réunion and the European Union, co-financed by the FEDER INTERREG 2021-2027 programme, No. DESVE/20251283, Ref. D2026/3446.
DM-O acknowledges funding from FONDECYT Iniciación 11250295. DM-O gratefully acknowledges support from the Centre for Biotechnology and Bioengineering (CeBiB; PIA project FB0001 and AFB240001, ANID, Chile). PEACCEL was supported through a research program partially co-funded by the European Union (EU) and Région Réunion (FEDER). C-F.K is supported through a research program funded by the European Union (EU) and Région Réunion (FEDER-FSE 2021/2027, n° 2025-0954-007180).

\section*{Declarations}
\textbf{Conflict of interest} The authors declare no conflict of interest.

\section*{Author contributions statement}

FC, MB, C-FK and DM-O: conceptualization. X.C: software implementation. FC, MB, C-FK, AM, JPC and DM-O: methodology. X.C, FC, MB, C-FK, AM, CD, JPC and DM-O: validation. X.C, FC, MB and C-FK: investigation. X.C, FC, MB, C-FK, AM, JPC and DM-O: writing and editing. FC, CD and MB: supervision and funding resources. FC and CD: project administration. All authors reviewed and approved the final version of the manuscript.

\section*{Artificial intelligence usage statement}

Artificial intelligence-based tools were used exclusively to assist with language refinement, grammatical correction, and text editing during manuscript preparation. All scientific content, methodological decisions, software development, analyses, and conclusions were conceived, implemented, verified, and validated by the authors.

\section*{Code and data availability statement}

No empirical dataset is required for the theoretical results presented in this manuscript. The illustrative scripts used to generate the schematic helical figures will be made available in a public repository upon acceptance or upon arXiv release.

\begin{appendices}

\section*{Appendix outline}

\begin{enumerate}
\item Appendix A. Biological and geometric motivation
\begin{enumerate}
    \item Empirical realization of the deformation sector
    \item Protein observables generating the connection field
\end{enumerate}

\item Appendix B. Technical Foundations of Admissible Transport
\begin{enumerate}
    \item Representative formulas for path operations
    \item Additional properties of normalized unitary cocycle
    \item Additional properties of ordered transport algebra
    \item Additional properties of commutative shadow
    \item Additional properties of transport core
\end{enumerate}

\item Appendix C. Proof of propositions and theorems
\begin{enumerate}
    \item Proof of Proposition~\ref{prop:global-geometric-cocycle}
    \item Proof of Proposition~\ref{prop:canonical-smooth-representation}
    \item Proof of Theorem~\ref{thm:selfadjoint-dirac}
    \item Proof of Theorem~\ref{thm:local-germ-invariance}
\end{enumerate}
\end{enumerate}

\section{Biological and geometric motivation}
This appendix is purely motivational and is not used in any theorem or construction of the main text.
\subsection{Empirical realization of the deformation sector}
\label{app:empirical-realization}
Empirical protein data enter the theory only at the level of geometric realization, not at the level of primitive definition.
Structural databases, NMR ensembles, and molecular-dynamics trajectories determine physically realized sectors of the abstract connection-grounded deformation framework.

Given an empirical input family
\[
X_{\mathrm{data}},
\]
the frame-to-connection extraction procedure induces a realization map
\[
\mathcal E_{\mathrm{data}}:X_{\mathrm{data}}\to \Sig.
\]
Its image defines the realized deformation sector
\[
\Sig_{\mathrm{data}}
:=
\mathcal E_{\mathrm{data}}(X_{\mathrm{data}})
\subset \Sig,
\]
together with the induced admissible path sector and the corresponding reduced ordered transport algebra. Numerical objects such as
\[
D_{\mathrm{num}},\qquad \mathfrak L^{\mathrm{num}}_\xi,\qquad \rho^{\mathrm{ren}}_{\mathrm{num}},\qquad \mathcal G_{\mathrm{num}}
\]
are then interpreted as approximations of the abstract constructions restricted to \(\Sig_{\mathrm{data}}\). Thus empirical protein data play a realizational and inferential role: they provide concrete sectors of the deformation-first connection model without altering its quaternionic foundation.

\subsection{Protein observables generating the connection field}
\label{app:protein-frame-transport}
Let
\[
x:I\to\R^3
\]
be a smooth immersed backbone curve parameterized by arc length.
Choose a smooth adapted orthonormal frame
\[
F(\ell)=(t(\ell),n(\ell),b(\ell))
\in SO(3),
\qquad \ell\in I,
\]
along the curve.

In nondegenerate regions one may use the Frenet frame, while nearly collinear regimes may instead be treated using a Bishop frame in order to avoid singular torsion behavior.

Since the interval \(I\) is contractible, the frame field admits a smooth spin lift
\[
q:I\to \SU(2),
\]
unique up to the global sign ambiguity
\(
q\sim -q
\).

Using the convention of Eq.~\eqref{eq:quaternionic-transport-field}, the lifted frame determines an infinitesimal transport field
\(
\Omega(\ell)\in\mathfrak{su}(2)\cong\operatorname{Im}\HH.
\)
Consequently, over a short arc-length increment \(\delta \ell\),
\[
q(\ell+\delta \ell)
\approx
q(\ell)
\exp\!
\Bigl(
\frac12\Omega(\ell)\delta \ell
\Bigr).
\]

Writing
\[
\Omega(\ell)
=
\omega_1(\ell)\mathbf i
+
\omega_2(\ell)\mathbf j
+
\omega_3(\ell)\mathbf k
\]
with respect to the chosen adapted frame identifies the local quaternionic transport parameters associated with the protein backbone geometry.

The norm
\(
|\Omega(\ell)|
\)
represents the instantaneous rotational rate, while its direction specifies the local rotation axis.
In the Frenet framing, the tangential component recovers the classical torsion function.
More generally, different adapted framings yield corresponding projection formulas determined by the chosen gauge convention.

Discrete realizations derived from residue-indexed frames or atomic-coordinate data are deferred to later computational implementations and are not required for the abstract geometric theory developed here.
\section{Technical Foundations of Admissible Transport}

\subsection{Representative formulas for path operations}
\label{app:path-operations}

Two piecewise \(C^1\) paths
\[
\gamma_1,\gamma_2:[0,1]\to\Sig
\]
are equivalent if
\[
\gamma_2=\gamma_1\circ\varphi
\]
for some orientation-preserving piecewise \(C^1\) bijection \(\varphi:[0,1]\to[0,1]\). When no confusion can arise, we write \(\gamma\) for the equivalence class \([\gamma]\). Its source and target are
\[
s(\gamma):=\gamma(0),\qquad t(\gamma):=\gamma(1).
\]
For \(\xi\in\Sig\), the constant class at \(\xi\) is denoted by \(\id_\xi\). The reversal of \(\gamma\) is represented by
\[
\gamma^{-1}(\tau):=\gamma(1-\tau),\qquad \tau\in[0,1].
\]
If \(t(\gamma_1)=s(\gamma_2)\), the ordered concatenation \(\gamma_2\circ\gamma_1\) is the class represented by
\[
(\gamma_2\circ\gamma_1)(\tau):=
\begin{cases}
\gamma_1(2\tau), & \tau\in[0,\frac12],\\[1mm]
\gamma_2(2\tau-1), & \tau\in[\frac12,1].
\end{cases}
\]
At this stage a deformation path is not reduced to its endpoints. Ordered concatenation belongs to the primary transport structure.

\subsection{Additional properties of normalized unitary cocycle}
\label{app:cocycle-axioms}

Let
\[
\Gad(\Sig)^{(2)}
:=
\{
(\gamma_2,\gamma_1)\in
\Gad(\Sig)\times\Gad(\Sig)
:
s(\gamma_2)=t(\gamma_1)
\},
\]
and
\[
\Gad(\Sig)^{(3)}
:=
\{
(\gamma_3,\gamma_2,\gamma_1)\in
\Gad(\Sig)^3
:
s(\gamma_3)=t(\gamma_2),
\quad
s(\gamma_2)=t(\gamma_1)
\}.
\]
A normalized unitary cocycle on \(\Gad(\Sig)\) is a map
\[
\sigma:\Gad(\Sig)^{(2)}\to U(1)
\]
satisfying
\[
\sigma(\gamma_3,\gamma_2\circ\gamma_1)
\,\sigma(\gamma_2,\gamma_1)
=
\sigma(\gamma_3\circ\gamma_2,\gamma_1)
\,\sigma(\gamma_3,\gamma_2)
\]
for every composable triple
\[
(\gamma_3,\gamma_2,\gamma_1)\in\Gad(\Sig)^{(3)},
\]
together with the normalization identities
\[
\sigma(\id_{t(\gamma)},\gamma)=1,
\qquad
\sigma(\gamma,\id_{s(\gamma)})=1.
\]


\subsection{Additional properties of ordered transport algebra}
\begin{enumerate}
\item \(\Atr(\Sig)\) is algebraically generated by the coefficient sector
\(C_c^\infty(\Sig)\) together with the transport generators \(U_\gamma\);

\item the involution reverses transport classes:
\[
U_{\gamma^{-1}}=U_\gamma^*;
\]

\item whenever \((\gamma_2,\gamma_1)\) is composable and
\(\gamma_2\circ\gamma_1\) is still a non-unit admissible class, one has
\[
U_{\gamma_2}U_{\gamma_1}
=
\sigma(\gamma_2,\gamma_1)\,U_{\gamma_2\circ\gamma_1};
\]

\item there exists a map
\[
\kappa:\Gad(\Sig)^{(2)}_{\mathrm{ret}}\to C_c^\infty(\Sig)
\]
such that whenever \((\gamma_2,\gamma_1)\in \Gad(\Sig)^{(2)}_{\mathrm{ret}}\), one has
\[
U_{\gamma_2}U_{\gamma_1}=\kappa(\gamma_2,\gamma_1).
\]
In particular, unit-return products are recorded in the coefficient sector, but their concrete realization is not fixed at the abstract level of the present definition;

\item the coefficient sector acts covariantly on transport generators:
\[
f\,U_\gamma=f(t(\gamma))\,U_\gamma,
\qquad
U_\gamma\,f=f(s(\gamma))\,U_\gamma,
\qquad f\in C_c^\infty(\Sig).
\]
\end{enumerate}
At this structural stage, no pointwise identity \(1_\xi\) is introduced, and no relation of the form
\[
U_{\id_\xi}=1_\xi
\]
is imposed.

\subsection{Additional properties of commutative shadow}
A commutative shadow of \(\Atr(\Sig)\) is a \(C_c^\infty(\Sig)\)-bimodule projection onto the coefficient algebra, namely
\[
E(fag)=f\,E(a)\,g,
\qquad
f,g\in C_c^\infty(\Sig),\ a\in \Atr(\Sig),
\]
\[
E(f)=f,
\qquad f\in C_c^\infty(\Sig),
\]
\[
E^2=E,
\qquad
E(a^*)=E(a)^*,
\qquad a\in \Atr(\Sig).
\]

\subsection{Additional properties of transport core}
Let
\[
b_1=(U_1,\beta_1),
\qquad
b_2=(U_2,\beta_2),
\]
and define
\[
U_{21}:=U_1\cap \beta_1^{-1}(U_2).
\]
If \(U_{21}=\varnothing\), set
\[
(U_{b_2}f_2)(U_{b_1}f_1):=0.
\]
If \(U_{21}\neq\varnothing\), define the composed local bisection
\[
b_2\circ b_1:=(U_{21},\beta_2\circ\beta_1)
\]
and the coefficient function
\[
\sigma[b_2,b_1](x)
:=
\sigma\bigl(b_2(\beta_1(x)),\,b_1(x)\bigr),
\qquad x\in U_{21}.
\]
Then, for
\[
f_1\in C_c^\infty(U_1),
\qquad
f_2\in C_c^\infty(U_2),
\]
the product is defined on generators by
\[
(U_{b_2}f_2)(U_{b_1}f_1)
=
U_{b_2\circ b_1}\Bigl(\sigma[b_2,b_1]\,(f_2\circ\beta_1)\,f_1\Bigr),
\]
with the understanding that if \(\beta_2\circ\beta_1=\mathrm{id}_{U_{21}}\), the right-hand side is interpreted through the identity relations as the zero extension of
\[
\sigma[b_2,b_1]\,(f_2\circ\beta_1)\,f_1\in C_c^\infty(U_{21}).
\]
The involution is defined on generators by
\[
(U_b f)^*
:=
U_{b^{-1}}\,(\overline f\circ \beta^{-1}),
\qquad
b^{-1}:=(\beta(U),\beta^{-1}),
\]
and extended anti-linearly to \(\Atr^\infty(\Sig)\). The coefficient algebra \(C_c^\infty(\Sig)\) is thus embedded as the distinguished commutative sector of \(\Atr^\infty(\Sig)\).

\section{Proof of propositions and theorems}
\subsection{Proof of proposition~\ref{prop:global-geometric-cocycle}}
\label{app:proof-global-cocycle}
Let $(\gamma_2,\gamma_1)\in \Gad(\Sig)^{(2)}$. Choose representatives of $\gamma_1$ and $\gamma_2$ that are piecewise geodesic for the fixed metric $g$ and subordinate to the fixed good geodesic cover. Choose moreover a geodesic triangulation of a filling surface for the ordered concatenation, with every $2$-simplex contained in one geodesically convex chart of the cover. Define $\sigma(\gamma_2,\gamma_1)$ as the product of the local phase factors supplied by Definition~\ref{def:geometric-cocycle} over the oriented triangles of such a triangulation.

If $(\gamma_2,\gamma_1)$ is already $U$-elementary, one may take the one-triangle filling in $U$, so the above construction gives exactly $\sigma_U(\gamma_2,\gamma_1)$. Thus the global map extends the local cocycle.

If one changes the subdivision or the admissible triangulation, the ratio of the two resulting phase factors is the exponential of the integral of $\Theta$ over a closed piecewise geodesic $2$-cycle $C\subset \Sig$. Because
\[
\frac{1}{2\pi}[\Theta]\in H^2(\Sig;\mathbb Z),
\]
one has
\[
\int_C \Theta \in 2\pi \mathbb Z,
\]
and therefore
\[
\exp\!\Bigl(i\int_C \Theta\Bigr)=1.
\]
Hence $\sigma(\gamma_2,\gamma_1)$ is independent of all auxiliary choices and takes values in $U(1)$ by construction.

Normalization follows from the fact that if one of the paths is a unit arrow, then the relevant filling can be chosen degenerate, so every local phase factor is $1$. For a composable triple $(\gamma_3,\gamma_2,\gamma_1)$, the two products
\[
\sigma(\gamma_3,\gamma_2\circ\gamma_1)\,\sigma(\gamma_2,\gamma_1)
\qquad\text{and}\qquad
\sigma(\gamma_3\circ\gamma_2,\gamma_1)\,\sigma(\gamma_3,\gamma_2)
\]
correspond to two geodesic triangulations of the same ordered filling attached to the triple. Their ratio is again the exponential of the integral of $\Theta$ over a closed piecewise geodesic $2$-cycle, hence equals $1$ by the same integrality argument. Therefore the cocycle identity holds, and $\sigma$ is a normalized unitary cocycle on $\Gad(\Sig)$.
\subsection{Proof of proposition~\ref{prop:canonical-smooth-representation}}
Let
\[
b=(U,\beta)
\]
be an admissible local bisection. By Definition~\ref{def:transport-operator-b}, each fiber map
\[
\widetilde\tau_{b,x}:\mathcal V_x\to \mathcal V_{\beta(x)}
\]
is unitary. Because $\beta:U\to \beta(U)$ is a diffeomorphism and the half-density factor is transported canonically, $V_b$ preserves the $L^2$-norm on sections supported in $U$ and maps them onto sections supported in $\beta(U)$. Hence $V_b$ extends to a bounded partial isometry on $\Hh$. Since $M_{\widetilde f}$ is bounded, each operator
\[
V_b M_{\widetilde f}
\]
is bounded on $\Hh$.

We first check that the assignment on generators is compatible with the defining relations of $\Atr^\infty(\Sig)$. If $U'\subset U$ is open and $\operatorname{supp}(f)\subset U'$, then the restriction $b|_{U'}=(U',\beta|_{U'})$ defines the same local connector family on the support of $f$, so
\[
V_b M_{\widetilde f}=V_{b|_{U'}}M_{\widetilde f}.
\]
Thus the restriction relations are respected. For the identity bisection $e_U=(U,\mathrm{id}_U)$, the connector at each $x\in U$ is constant, hence $\widetilde\tau_{e_U,x}=\mathrm{id}_{\mathcal V_x}$ and $V_{e_U}$ acts as the identity on sections supported in $U$. Therefore
\[
V_{e_U}M_{\widetilde f}=M_{\widetilde f},
\]
so the identity relations are also respected. The assignment therefore descends to a well-defined linear map on $\Atr^\infty(\Sig)$.

We next verify multiplicativity on generators. Let
\[
b_1=(U_1,\beta_1),\qquad b_2=(U_2,\beta_2),
\]
with
\[
f_1\in C_c^\infty(U_1),\qquad f_2\in C_c^\infty(U_2),
\]
and define
\[
U_{21}:=U_1\cap \beta_1^{-1}(U_2).
\]
If $U_{21}=\varnothing$, then $\beta_1(\operatorname{supp}f_1)$ is disjoint from $\operatorname{supp}f_2$, hence
\[
(V_{b_2}M_{\widetilde f_2})(V_{b_1}M_{\widetilde f_1})=0
=\pi_{\mathrm{can}}\bigl((U_{b_2}f_2)(U_{b_1}f_1)\bigr).
\]
Assume now that $U_{21}\neq\varnothing$ and write
\[
b_2\circ b_1=(U_{21},\beta_2\circ\beta_1).
\]
Let
\[
\psi\in C_c^\infty(\Sig,\mathcal V),\qquad x\in U_{21},\qquad y=(\beta_2\circ\beta_1)(x).
\]
Then
\[
(V_{b_1}M_{\widetilde f_1}\psi)(\beta_1(x))=\widetilde\tau_{b_1,x}\bigl(f_1(x)\psi(x)\bigr),
\]
and hence
\[
\begin{aligned}
\bigl((V_{b_2}M_{\widetilde f_2})(V_{b_1}M_{\widetilde f_1})\psi\bigr)(y)
&=\widetilde\tau_{b_2,\beta_1(x)}\Bigl(f_2(\beta_1(x))\,\widetilde\tau_{b_1,x}\bigl(f_1(x)\psi(x)\bigr)\Bigr)\\
&=\widetilde\tau_{b_2,\beta_1(x)}\widetilde\tau_{b_1,x}\Bigl((f_2\circ\beta_1)(x)f_1(x)\psi(x)\Bigr).
\end{aligned}
\]
By the cocycle identity encoded by the geometric phase construction,
\[
\widetilde\tau_{b_2,\beta_1(x)}\widetilde\tau_{b_1,x}
=\sigma[b_2,b_1](x)\,\widetilde\tau_{b_2\circ b_1,x}.
\]
Therefore
\[
\bigl((V_{b_2}M_{\widetilde f_2})(V_{b_1}M_{\widetilde f_1})\psi\bigr)(y)
=\bigl(V_{b_2\circ b_1}M_{\widetilde h}\psi\bigr)(y),
\]
where
\[
h:=\sigma[b_2,b_1]\,(f_2\circ\beta_1)f_1\in C_c^\infty(U_{21}).
\]
Thus
\[
(V_{b_2}M_{\widetilde f_2})(V_{b_1}M_{\widetilde f_1})
=V_{b_2\circ b_1}M_{\widetilde h}.
\]
If $\beta_2\circ\beta_1\neq \mathrm{id}_{U_{21}}$, this is exactly
\[
\pi_{\mathrm{can}}\bigl((U_{b_2}f_2)(U_{b_1}f_1)\bigr).
\]
If $\beta_2\circ\beta_1=\mathrm{id}_{U_{21}}$, then $b_2\circ b_1$ is the identity bisection on $U_{21}$, so by the identity relation already checked,
\[
V_{b_2\circ b_1}M_{\widetilde h}=M_{\widetilde h}=
\pi_{\mathrm{can}}\bigl((U_{b_2}f_2)(U_{b_1}f_1)\bigr).
\]
Hence $\pi_{\mathrm{can}}$ is multiplicative.

We next verify the involution. Let
\[
b=(U,\beta),\qquad f\in C_c^\infty(U).
\]
Since each fiber transport is unitary, one has $V_b^*=V_{b^{-1}}$. Therefore
\[
\pi_{\mathrm{can}}(U_b f)^*=(V_bM_{\widetilde f})^*=M_{\widetilde f}^{\,*}V_b^*=M_{\widetilde{\overline f}}V_{b^{-1}}.
\]
For every test section $\psi$ and every $x\in U$ one checks directly that
\[
M_{\widetilde{\overline f}}V_{b^{-1}}\psi
=V_{b^{-1}}M_{\widetilde{\overline f\circ \beta^{-1}}}\psi.
\]
Hence
\[
\pi_{\mathrm{can}}(U_b f)^*=V_{b^{-1}}M_{\widetilde{\overline f\circ \beta^{-1}}}
=\pi_{\mathrm{can}}\bigl((U_b f)^*\bigr).
\]
So $\pi_{\mathrm{can}}$ is a $*$-representation.

For the global identity bisection
\[
e=(\Sig,\mathrm{id}),
\]
one has $V_e=\mathrm{Id}_{\Hh}$, so
\[
\pi_{\mathrm{can}}(U_e f)=M_f
\]
for all $f\in C_c^\infty(\Sig)$. Since multiplication by compactly supported smooth functions acts nondegenerately on $\Hh$, the representation $\pi_{\mathrm{can}}$ is nondegenerate.

Finally, both multiplication by a compactly supported smooth coefficient and the operator $V_b$ preserve smooth compact support. Therefore the dense subspace
\[
C_c^\infty(\Sig,\mathcal V)
\]
is invariant under
\[
\pi_{\mathrm{can}}\bigl(\Atr^\infty(\Sig)\bigr).
\]

\subsection{Proof of theorem~\ref{thm:selfadjoint-dirac}}
The bundle
\[
\mathcal V=S\otimes L\otimes |\Lambda|^{1/2}
\]
is Hermitian, and the connection
\[
\nabla^{\mathcal V}
\]
is the tensor-product connection induced by the spin$^c$ connection on \(S\), the unitary connection
\(\nabla^L\) on \(L\), and the Levi-Civita connection on \(|\Lambda|^{1/2}|\). Hence
\[
\mathfrak D_0=c_0\circ \nabla^{\mathcal V}
\]
is a formally self-adjoint Dirac-type operator associated with a Hermitian Clifford connection.

For
\[
\varphi,\psi\in C_c^\infty(\Sig,\mathcal V),
\]
the standard Green formula for Dirac-type operators gives
\[
\langle \mathfrak D_0\varphi,\psi\rangle_{\Hh}
=
\langle \varphi,\mathfrak D_0\psi\rangle_{\Hh},
\]
because the sections are compactly supported and therefore no boundary term appears. Thus
\(\mathfrak D_0\) is symmetric on
\[
C_c^\infty(\Sig,\mathcal V).
\]

By Definition~\ref{def:canonical-prototype}, the Riemannian manifold
\[
(\Sig,g)
\]
is complete. The standard completeness theorem for formally self-adjoint Dirac-type operators on complete
Riemannian manifolds therefore applies, and yields that \(\mathfrak D_0\) is essentially self-adjoint on
\[
C_c^\infty(\Sig,\mathcal V).
\]
Let \(\mathfrak D\) denote its closure. Then \(\mathfrak D\) is self-adjoint, and it is the unique self-adjoint
realization of the same first-order Dirac-type differential expression.

Now let
\[
f\in C_c^\infty(\Sig),
\qquad
\psi\in C_c^\infty(\Sig,\mathcal V).
\]
By Proposition~\ref{prop:canonical-smooth-representation}, the coefficient element \(f\) is represented by
\[
\pi_{\mathrm{can}}(f)=M_f,
\]
multiplication by \(f\). Since \(f\) is smooth with compact support, \(M_f\) preserves
\(C_c^\infty(\Sig,\mathcal V)\).

Because \(\nabla^{\mathcal V}\) is a Clifford connection, the Leibniz rule gives
\[
\mathfrak D_0(f\psi)=c_0(df)\psi + f\,\mathfrak D_0\psi.
\]
Therefore
\[
[\mathfrak D_0,M_f]\psi = c_0(df)\psi.
\]
Since \(\mathfrak D\) extends \(\mathfrak D_0\), the same identity holds on the core
\(C_c^\infty(\Sig,\mathcal V)\):
\[
[\mathfrak D,\pi_{\mathrm{can}}(f)]\psi = c_0(df)\psi.
\]
Under the metric identification
\[
d_{\EP}f=(df)^\sharp_g\in \EP,
\]
this is exactly the Hilbert-space Clifford action of \(d_{\EP}f\), so
\[
[\mathfrak D,\pi_{\mathrm{can}}(f)] = c_{\Hh}(d_{\EP}f)
\]
on \(C_c^\infty(\Sig,\mathcal V)\).

Finally, because \(df\) has compact support, fiberwise Clifford multiplication by \(df\) defines a bounded
bundle endomorphism of \(\mathcal V\), hence a bounded operator on \(\Hh\).

\subsection{Proof of theorem~\ref{thm:local-germ-invariance}}
The tangent-groupoid deformation of a first-order differential operator is intrinsic: it depends only on
the operator and on the smooth structure of the underlying manifold. Applied to
\[
\mathfrak D_0=c_0\circ \nabla^{\mathcal V},
\]
this produces a well-defined germ
\[
\mathfrak L_\xi(\mathfrak D)
\]
along \(T_\xi \Sig\times\{0\}\), independently of any local presentation.

To describe concrete representatives of this germ, choose a geodesic normal coordinate chart centered at
\(\xi\) together with a compatible unitary trivialization of \(\mathcal V\). In such a presentation,
\(\mathfrak D_0\) is represented by a first-order Dirac-type differential operator whose principal
symbol at \(\xi\) is the Clifford multiplication by covectors. The associated normal operator is the
constant-coefficient tangent Dirac operator obtained from that symbol. If one changes the normal chart,
the derivative of the transition map at \(\xi\) is orthogonal with respect to \(g_\xi\); if one changes
the compatible trivialization, the new frame differs by a fiberwise unitary transformation of
\(\mathcal V_\xi\). Together these changes induce a canonical unitary intertwiner on
\[
L^2(T_\xi \Sig,\mathcal V_\xi),
\]
so any two representatives of \(\mathcal N_\xi(\mathfrak D_0)\) are unitarily equivalent. This proves
that the tangent operator is canonical only up to that natural unitary equivalence, while the germ
itself is intrinsically well defined.

Since each representative \(D(\xi)\) is a constant-coefficient symmetric Dirac-type operator on the
Euclidean vector space \(T_\xi \Sig\) endowed with the inner product \(g_\xi\), Fourier transform
identifies it with multiplication by a Hermitian matrix-valued symbol. The standard Fourier analysis then
shows that \(D(\xi)\) is symmetric and essentially self-adjoint on
\[
C_c^\infty(T_\xi \Sig,\mathcal V_\xi).
\]

\end{appendices}


\bibliography{sn-bibliography}

\end{document}